\documentclass{article}
\usepackage{amsfonts,amssymb,amsmath,a4wide,paralist}
\usepackage{epsfig,rotating,framed,enumerate,listings}
\usepackage{color}
\usepackage{paralist}
\usepackage{marginnote,nicefrac}
\usepackage{authblk}
\usepackage{tikz}
\usepackage{xfrac}

\usepackage[paper=a4paper,left=31mm,right=31mm,top=30mm, bottom=25mm]{geometry}

\usepackage{hyperref}

\newcommand{\NN}{{\mathbb N}}
\newcommand{\ZZ}{{\mathbb Z}}
\newcommand{\bigo}{\ensuremath{\mathcal{O}}}

\newcommand{\dpw} {\text{d-pw}}

\newcommand{\ideg}{\text{indegree}}
\newcommand{\odeg}{\text{outdegree}}

\newcommand{\sproof}{\noindent{\bf Proof}\quad}
\newcommand{\eproof}{\hfill $\Box$}

\newcommand{\FI}{{\it first}}
\newcommand{\LA}{{\it last}}
\newcommand{\gansfuss}[1]{\mbox{``}{#1}\mbox{''}}
\newcommand{\pw} {\text{pw}}
\newcommand{\g} {{\it g}}
\newcommand{\q} {{\it q}}
\newcommand{\s} {{\it s}}
\newcommand{\un} {{\it und}}
\newcommand{\ac}{{\it active}}
\newcommand{\co} {\text{co-}}
\newcommand{\con} {\text{con-}}
\newcommand{\types}{{\it types}}
\newcommand{\ACT}{{\it active}}

\newtheorem{theorem}{Theorem}[section]
\newtheorem{example}[theorem]{Example}
\newtheorem{lemma}[theorem]{Lemma}
\newtheorem{definition}[theorem]{Definition}
\newtheorem{proposition}[theorem]{Proposition}
\newtheorem{observation}[theorem]{Observation}
\newtheorem{corollary}[theorem]{Corollary}
\newtheorem{claim}[theorem]{Claim}

\newtheorem{remark}[theorem]{Remark}

\newenvironment{proof}{\noindent{\bf Proof~}}{\null\hfill $\Box$\par\medskip}

\definecolor{light-gray}{gray}{0.5}

\begin{document}

\title{Characterizations and Directed Path-Width of Sequence Digraphs\thanks{A short version of 
this paper appeared in Proceedings of the {\em 12th Annual International Conference on Combinatorial Optimization and Applications} (COCOA 2018) \cite{GRR18}.}}

\author[1]{Frank Gurski}
\author[1]{Carolin Rehs}

\affil[1]{\small University of  D\"usseldorf,
Institute of Computer Science, Algorithmics for Hard Problems Group,\newline 
40225 D\"usseldorf, Germany}

\author[2]{Jochen Rethmann}

\affil[1]{\small Niederrhein University of Applied Sciences,
Faculty of Electrical Engineering and Computer Science,
47805 Krefeld, Germany}

\maketitle

\begin{abstract}
Computing the directed path-width of a directed graph is an NP-hard problem. 
Even for  digraphs of maximum semi-degree 3 the problem remains hard.  
We propose a decomposition of an input digraph $G=(V,A)$ by a number $k$ of sequences with entries 
from $V$, such that $(u,v)\in A$  if and only if in one of the sequences there is an 
occurrence of $u$ appearing before an occurrence of $v$. We present several graph theoretical
properties of these digraphs. Among these we give forbidden subdigraphs of digraphs which 
can be defined by $k=1$  sequence, which is a subclass of semicomplete digraphs. 
Given the decomposition of digraph $G$, we show an algorithm which computes the 
directed path-width of $G$ in time $\bigo(k\cdot (1+N)^k)$, where $N$ denotes the
maximum sequence length. This leads to an XP-algorithm w.r.t.\ $k$ for 
the directed path-width problem. Our result improves the algorithms of Kitsunai et al. 
for digraphs of large directed path-width which can be decomposed by a small number of sequence.

\end{abstract}

\section{Introduction}

Let $Q=\{q_1,\ldots,q_k\}$ be a set of $k$ sequences $q_i=(b_{i,1},\ldots,b_{i,n_i})$, $1\leq i \leq k$.
All $n_1+\ldots+n_k$ items $b_{i,j}$ are pairwise distinct. 
Further there is a function $t$ which assigns every item $b_{i,j}$ with a type $t(b_{i,j})$.
The sequence digraph $\g(Q) = (V,A)$ for a set $Q=\{q_1, \ldots ,q_k\}$
has a vertex for every type and an arc  $(u,v) \in A$ if and only if there is some sequence 
$q_i$ in $Q$ where an item of type $u$ is on the left of some item of type $v$.\footnote{A related but different
concept for undirected graphs is the notation of word-representable graphs. 
A graph $G =(V, E)$ is {\em word-representable} if there exists a word $w$
over the alphabet $V$ such that letters $x$ and $y$ alternate in $w$ if and
only if $\{x,y\}\in E$., see \cite{KP18} for a survey.}
%
The contributions of this paper concern graph theoretic properties of sequence digraphs
and an algorithm to compute from a given some set $Q$ on $k$ sequences the 
directed path-width of $\g(Q)$. 

This paper is organized as follows. 
In Section \ref{sec-pre}
we give preliminaries for graphs and digraphs. 
In Section \ref{sequencegraphsandsys}
we show how to define digraphs form sets of sequences, and vice versa.
Further we give methods in order to compute the sequence digraph and also its complement digraph.
In Section \ref{properties} we consider graph theoretic properties of sequence digraphs.
Therefore in Section \ref{classes} we introduce the class $S_{k,\ell}$
as the set of all sequence digraphs defined by sets $Q$ 
on at most $k$ sequences that together contain at most $\ell$ items of each type.
We show inclusions of these classes and necessary conditions 
for digraphs to be in $S_{k,\ell}$.
In Section \ref{sec-char-1s}  we give finite sets of  forbidden induced subgraphs
for all classes  $S_{1,1}$ and  $S_{1,2}$.
It turns out that $S_{1,1}$ is equal to the well known class of transitive
tournaments. Since only the first and the last item of each type in every $q_i\in Q$ are important
for the arcs in the  corresponding digraph all classes $S_{1,\ell}$, $\ell\geq 2$ are equal and
can be characterized by one set of four forbidden induced subdigraphs.
These characterizations lead to polynomial time  recognition algorithms for 
the corresponding graph classes. Furthermore we give characterizations
in terms of special tournaments and conditions for the complement digraph.
In Section \ref{sec-co} we consider the relations of sequence digraphs to directed co-graphs 
(defined in  \cite{BGR97})
and  their subclass oriented threshold graphs (defined in \cite{Boe15}). 
In Section \ref{sec-converse} we consider the closure of the classes $S_{k,\ell}$
with respect to the operations taking the converse digraph and taking the 
complement digraph.
In Section \ref{sec-sim} we characterize sequences whose defined
graphs can be obtained by the union of the graphs for the subsequences
which consider only  the first or last item of each type.

In Section \ref{sec-problems} we 
consider the directed path-width problem on sequence digraphs.
We show that for digraphs defined by $k=1$ sequence the directed path-width 
can be computed in polynomial time. Further we show that for sets $Q$
of sequences of bounded length, of bounded distribution of the items of every type 
onto the sequences, or bounded number of items of every type
computing the directed path-width of $\g(Q)$ is NP-hard.
We also introduce a method which computes 
from a given set $Q$ on $k$ sequences the 
directed path-width of $\g(Q)$, as well as a directed path-decomposition 
can be computed in time 
$\bigo(k\cdot(1+\max_{1\leq i \leq k}n_i)^k)$.
The main idea is to discover an optimal directed path-decomposition 
by scanning the $k$ sequences left-to-right and keeping in a state 
the numbers of scanned items of every sequence and the number of active types.

From a parameterized point of view our method leads to an XP-algorithm w.r.t.~parameter $k$.
While the existence of FPT-algorithms for computing directed path-width is open up to now,
there are  XP-algorithms for the directed path-width problem for some digraph $G=(V,A)$. 
The directed path-width can be computed in time 
$\bigo(\nicefrac{|A|\cdot |V|^{2\dpw(G)}}{(\dpw(G)-1)!})$ 
by \cite{KKKTT16} and in time 
$\bigo(\dpw(G)\cdot|A|\cdot |V|^{2\dpw(G)})$ 
by \cite{Nag12}. Further in \cite{KKT15} it is shown how to decide whether the directed
path-width of an $\ell$-semicomplete digraph is at most $w$ 
in time $(\ell+2w+1)^{2w}\cdot |V|^{\bigo(1)}$.
All these algorithms are exponential in the directed path-width of the
input digraph while our algorithm is exponential within the number of sequences.
Thus our result improves theses algorithms for digraphs of large directed 
path-width which can be decomposed by a small number of sequences (see Table \ref{F-co2x} for examples).
Furthermore the directed path-width can be computed in time 
$3^{\tau(\un(G))} \cdot |V|^{\bigo(1)}$, where  $\tau(\un(G))$ denotes the vertex
cover number of the underlying undirected graph of $G$, by \cite{Kob15}.
Thus our result improves this algorithm for digraphs of large $\tau(\un(G))$
which can be decomposed by a small number of sequences (see Table \ref{F-co2x} for examples).

\begin{table}[!ht]
$$
\begin{tabular}{|l|cccc|}
\hline
digraphs $G$    &  ~$\dpw(G)$~ &$\tau(\un(G))$ & ~~~$k$~~~  &  ~~~$\ell$~~~ \\
\hline
transitive tournaments  & 0 &  $n-1$ &  1 & 1 \\
union of $k'$ transitive tournaments   & 0 & $(\sum_{i=1}^{k'}n_i) -k'$ &  $k'$ & 1 \\
bidirectional complete digraphs $\overleftrightarrow{K_n}$ & $n-1$ & $n-1$ & 1 & 2 \\ 
semicomplete $\{\overrightarrow{C_3},D_0,D_4\}$-free & $[0,n-1]$ & $n-1$  & 1 &2\\
semicomplete $\{\co(2\overrightarrow{P_2}),\overrightarrow{C_3},D_4\}$-free & $[0,n-1]$ & $n-1$  & 1 &2\\
union of $k'$ semicomplete $\{\co(2\overrightarrow{P_2}),\overrightarrow{C_3},D_4\}$-free &$[0,n-1]$ &  $(\sum_{i=1}^{k'}n_i) -k'$& $k'$ & $2k'$\\
\hline
\end{tabular}
$$
\caption{Values of parameters within XP-algorithms for directed path-width.}
\label{F-co2x}
\end{table}

Since every single sequence $q_i$ defines 
a semicomplete $\{\co(2\overrightarrow{P_2}),\overrightarrow{C_3},D_4\}$-free digraph (cf.
Table \ref{F-co2} for the digraphs), we 
consider digraphs which can be obtained by the union of $k$ special semicomplete digraphs
which confirms the conjecture of \cite{KKT15} that semicompleteness
is a useful restriction when considering digraphs.\footnote{When considering the directed path-width of almost semicomplete digraphs
in \cite{KKT15} the class
of semicomplete digraphs was suggested to be \gansfuss{a promising stage for pursuing digraph analogues of
the splendid outcomes, direct and indirect, from the Graph Minors project}.}
In Section \ref{sec-concl}
we give conclusions and point out open questions.

\section{Preliminaries}\label{sec-pre}

We use the notations of Bang-Jensen and Gutin \cite{BG09} for graphs and digraphs.

\paragraph{Undirected graphs}
We work with finite undirected {\em graphs} $G=(V,E)$,
where $V$ is a finite set of {\em vertices} 
and $E \subseteq \{ \{u,v\} \mid u,v \in
V,~u \not= v\}$ is a finite set of {\em edges}.
For a vertex $v\in V$ we denote by $N_G(v)$ 
the set of all vertices which are adjacent to $v$ in $G$, 
i.e.~$N_G(v)=\{w\in V~|~\{v,w\}\in E\}$. 
Set $N_G(v)$ is called the set of all {\em neighbors} of $v$ 
in $G$ or {\em neighborhood} of $v$ in $G$.  
The {\em degree} of a vertex $v\in V$,  denoted by $\deg_G(v)$, 
is the number of neighbors of vertex $v$ in $G$, i.e.~$\deg_G(v)=|N_G(v)|$. 
The maximum vertex degree is $\Delta(G)=\max_{v\in V} \deg_G(v)$.
A graph $G'=(V',E')$ is a {\em subgraph} of graph $G=(V,E)$ if $V'\subseteq V$ 
and $E'\subseteq E$.  If every edge of $E$ with both end vertices in $V'$  is in
$E'$, we say that $G'$ is an {\em induced subgraph} of digraph $G$ and 
we write $G'=G[V']$.
For some undirected graph $G=(V,E)$ its complement graph is defined by
$\co G=(V,\{\{u,v\}~|~\{u,v\}\not\in E, u,v\in V, u\neq v\})$.

\paragraph{Special Undirected Graphs}
By $P_n=(\{v_1,\ldots,v_n\},\{\{v_1,v_2\},\ldots, \{v_{n-1},v_n\}\})$, $n \ge 2$,
we denote a path on $n$ vertices and by  $C_n=(\{v_1,\ldots,v_n\},\{\{v_1,v_2\},\ldots, \{v_{n-1},v_n\},\{v_n,v_1\}\})$,
$n \ge 3$, we denote a cycle on $n$ vertices.
Further by $K_n=(\{v_1,\ldots,v_n\},\{\{v_i,v_j\}~|~1\leq i<j\leq n\})$,
$n \ge 1$, we denote a complete graph on $n$ vertices and by $K_{n,m}=(\{v_1,\ldots,v_n,w_1,\ldots,w_m\},
\{\{v_i,w_j\}~|~ 1\leq i\leq n, 1\leq j \leq m\})$  a complete bipartite graph on $n+m$ vertices.

\paragraph{Directed graphs}
A {\em directed graph} or {\em digraph} is a pair  $D=(V,A)$, where $V$ is 
a finite set of {\em vertices} and  $A\subseteq \{(u,v) \mid u,v \in
V,~u \not= v\}$ is a finite set of ordered pairs of distinct\footnote{Thus we do not consider 
directed graphs with loops.} vertices called {\em arcs}. 
For a vertex $v\in V$, the sets $N_D^+(v)=\{u\in V~|~ (v,u)\in A\}$ and 
$N_D^-(v)=\{u\in V~|~ (u,v)\in A\}$ are called the {\em set of all out-neighbours}
and the {\em set of all in-neighbours} of $v$. 
The  {\em outdegree} of $v$, $\odeg_D(v)$ for short, is the number
of out-neighbours of $v$ and the  {\em indegree} of $v$, $\ideg_D(v)$ for short, 
is the number of in-neighbours of $v$ in $D$.
The {\em maximum out-degree} is $\Delta^+(D)=\max_{v\in V} \odeg_D(v)$ and
the {\em maximum in-degree} is $\Delta^-(D)=\max_{v\in V} \ideg_D(v)$.
The {\em maximum vertex degree} is $\Delta(G)=\max_{v\in V} \odeg_G(v)+\ideg_G(v)$
and the  {\em maximum semi-degree} is $\Delta^0(G)=\max\{\Delta^-(D),\Delta^+(D)\}$.

A vertex  $v\in V$ is  {\em out-dominating (in-dominated)} if
it is adjacent to every other vertex in $V$ and is a source (a sink, respectively).
A digraph $D'=(V',A')$ is a {\em subdigraph} of digraph $D=(V,A)$ if $V'\subseteq V$ 
and $A'\subseteq A$.  If every arc of $A$ with both end vertices in $V'$  is in
$A'$, we say that $D'$ is an {\em induced subdigraph} of digraph $D$ and we 
write $D'=D[V']$.
For some digraph $D=(V,A)$ its {\em complement digraph} is defined by
$\co D =(V,\{(u,v)~|~(u,v)\not\in A, u,v\in V, u\neq v\})$
and its {\em converse digraph} is defined by
$\con D=(V,\{(u,v)~|~(v,u)\in A, u,v\in V, u\neq v\})$.

%

Let $G=(V,A)$ be a digraph.

\begin{itemize}
\item $G$ is {\em edgeless} if for all $u,v \in V$, $u \neq v$, 
none of the two pairs $(u,v)$ and $(v,u)$ belongs to $A$.

\item $G$ is a {\em tournament} if for all $u,v \in V$, $u \neq v$, 
exactly one of the two pairs $(u,v)$ and $(v,u)$ belongs to $A$.

\item $G$ is {\em semicomplete} if for all $u,v \in V$, $u \neq v$, 
at least one of the two pairs $(u,v)$ and $(v,u)$ belongs to $A$.

\item $G$ is {\em (bidirectional) complete} if for all $u,v \in V$, $u \neq v$, 
both of the two pairs $(u,v)$ and $(v,u)$ belong to $A$.
\end{itemize}

\paragraph{Omitting the directions} 
For some given digraph $D=(V,A)$, we define
its underlying undirected graph by ignoring the directions of the 
edges, i.e.~$\un(D)=(V,\{\{u,v\}~|~(u,v)\in A, u,v\in V\})$.

\paragraph{Orientations} 
There are several ways to define a digraph $D=(V,A)$ from an undirected 
graph $G=(V,E)$, see \cite{BG09}.
If we replace every edge $\{u,v\}$ of $G$ by 
\begin{itemize}
\item
one of the arcs $(u,v)$ and $(v,u)$, we denote $D$ as an {\em orientation} of $G$.
Every digraph $D$  which can be obtained by an orientation of some undirected
graph $G$ is called an {\em oriented graph}.

\item
one or both of the arcs $(u,v)$ and $(v,u)$, we denote $D$ as a {\em biorientation} of $G$.
Every digraph $D$  which can be obtained by a biorientation of some undirected
graph $G$ is called a {\em bioriented graph}.

\item
both arcs $(u,v)$ and $(v,u)$, we denote $D$ as a {\em complete biorientation} of $G$.
Since in this case $D$ is well defined by $G$ we also denote
it by $\overleftrightarrow{G}$.
Every digraph $D$  which can be obtained by a complete biorientation of some undirected
graph $G$ is called a {\em complete bioriented graph}. 
\end{itemize}

\paragraph{Special Directed Graphs}We recall some special directed graphs.
We denote by $\overleftrightarrow{K_n}=(\{v_1,\ldots,v_n\},\{ (v_i,v_j)~|~1\leq i\neq j\leq n\})$,
$n \ge 1$ a bidirectional complete digraph on $n$ vertices.
By $\overrightarrow{P_n}=(\{v_1,\ldots,v_n\},\{ (v_1,v_2),\ldots, (v_{n-1},v_n)\})$, $n \ge 2$
we denote a directed path on $n$ vertices and by  
$\overrightarrow{C_n}=(\{v_1,\ldots,v_n\},\{(v_1,v_2),\ldots, (v_{n-1},v_n),(v_n,v_1)\})$, $n \ge 2$
we denote a directed cycle on $n$ vertices.

The {\em $k$-power graph} $G^k$ of a digraph $G$ is a graph with the 
same vertex set as $G$. There is an arc $(u,v)$ in $G^k$ if and only if there 
is a directed path from $u$ to $v$ of length at most $k$ in $G$. 

An {\em oriented forest (tree)} is the orientation of a forest (tree). A digraph is
an {\em out-tree (in-tree)} if it is an oriented tree in which there is
exactly one vertex of indegree (outdegree) zero. 

A {\em directed acyclic digraph (DAG for short)} is a digraph without any $\overrightarrow{C_n}$,
$n\geq 2$ as subdigraph.

\section{Sequence Digraphs and Sequence Systems}\label{sequencegraphsandsys}

\subsection{From Sequences to Digraphs}\label{SCgb} 

Let  $Q=\{q_1,\ldots,q_k\}$ be a set of $k$ sequences. 
Every sequence $q_i=(b_{i,1},\ldots,b_{i,n_i})$ consists of a number $n_i$ of items, such
that all $n=\sum_{i=1}^k n_i$ items are pairwise distinct. 
Further there is a function $t$ which assigns every item $b_{i,j}$ with a type $t(b_{i,j})$.
The set of all types of the items in some sequence $q_i$ is denoted by
$\types(q_i)= \{ t(b) ~|~ b \in q_i \}$.
For a set of sequences $Q = \{q_1, \ldots, q_k\}$ we denote
$\types(Q) = \types(q_1) \cup \cdots \cup \types(q_k)$.
For some sequence  $q_\ell=(b_{\ell,1}, \ldots,b_{\ell,n_{\ell}})$
we say item $b_{\ell,i}$ is {\it on the
left of} item $b_{\ell,j}$ in sequence $q_{\ell}$ if $i < j$. Item
$b_{\ell,i}$ is on the {\it position} $i$ in sequence $q_{\ell}$, since there are $i-1$
items on the left of $b_{\ell,i}$ in sequence $q_{\ell}$.  

In order to {\em insert} a new item $b$ on a position $j$ in
sequence $q_i$ we first move all items on positions $j'\geq j$ to position $j'+1$ starting at the
rightmost position $n_i$ and
then we insert $b$ at position $j$. In order to {\em remove}
an existing item $b$ at a position $j$ in sequence $q_i$ we move all items from
positions $j'\geq j+1$ to position $j'-1$ starting at position $j+1$.

For hardness results  in Section \ref{sec-h} we 
consider the distribution of the items of a type $t$
onto the sequences by
\[
   d_Q(t) = |\{ q\in Q \mid t \in \types(q) \}|
     \qquad \mbox{ and } \qquad
     d_Q = \max_{t \in \types(Q)} d_Q(t).
\]
For the number of items for type $t$ within the sequences we define
\[
   c_Q(t) = \sum_{q\in Q} |\{ b \in q \mid t(b)=t\}|
     \qquad \mbox{ and } \qquad
     c_Q = \max_{t \in \types(Q)} c_Q(t).
\]
Obviously it holds $d_Q\leq k$ and $1\leq d_Q\leq c_Q\leq n$.

\begin{definition}[Sequence Digraph]
The {\em sequence digraph} $\g(Q) = (V,A)$ for a set $Q=\{q_1, \ldots ,q_k\}$
has a vertex for every type, i.e.~$V=\types(Q)$ and an arc  $(u,v) \in A$ 
if and only if there is some sequence $q_i$ in $Q$
where an item of type $u$ is on the left of an item of type $v$.
More formally, there is an arc $(u,v) \in A$ if and only if there is some sequence
$q_i$ in $Q$, such that there are two items $b_{i,j}$ and $b_{i,j'}$ such that
\begin{enumerate}[(1)]
\item $1 \leq j < j' \leq n_i$,

\item $t(b_{i,j})=u$,

\item $t(b_{i,j'})=v$, and 

\item $u\neq v$.
\end{enumerate}
\end{definition}


Sequence digraphs have successfully been applied in order to model the stacking process of bins from
conveyor belts onto pallets with respect to customer orders, which is an important task in
palletizing systems used in centralized distribution centers \cite{GRW16b}. 
In this paper we show graph theoretic properties of sequence digraphs and their relation
to special digraph classes. 

In our examples we will use type identifications instead of
item identifications to represent a sequence $q_i\in Q$. For $r$ not
necessarily distinct types $t_1, \ldots, t_r$ let $[t_1,\ldots,t_r]$
denote some sequence $q_i=(b_{i,1}, \ldots, b_{i,r})$ of $r$ pairwise distinct items,
such that $t(b_{i,j}) = t_j$ for $j=1,\ldots,r$.
We use this notation for sets of sequences as well.

\begin{example}[Sequence Digraph]\label{EX6}
Figure \ref{F04} shows the sequence digraph $\g(Q)$ for $Q = \{q_1, q_2, q_3\}$ 
with sequences $q_1 = [a,a,d,e,d]$, $q_2 = [c,b,b,d]$, and 
$q_3 = [c,c,d,e,d]$.
\end{example}

\begin{figure}[hbtp]
\centering
\parbox[b]{.50\textwidth}{
\centerline{\epsfxsize=40mm \epsfbox{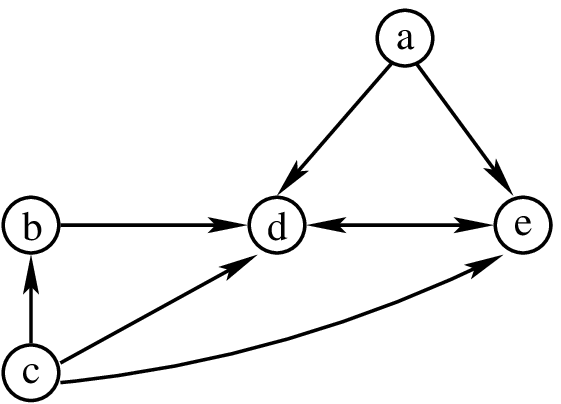}}
\caption{Sequence digraph $\g(Q)$ of Example \ref{EX6}.}
\label{F04}
}
\hfill
\parbox[b]{.48\textwidth}{
\centerline{\epsfxsize=45mm \epsfbox{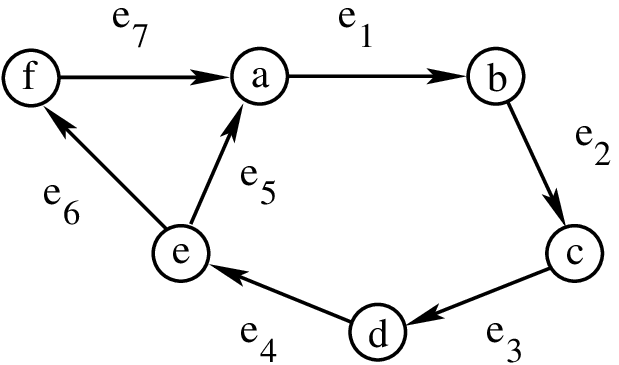}}
\caption{Digraph $G$ of Example \ref{EX3-new}.}
\label{F-ex-d}
}
\end{figure}

In Example \ref{EX6} all sequences of $Q$ contain consecutive items
of the same type. This is obviously not necessary to obtain digraph $\g(Q)$.
%
Let $C(q_i)$ be the subsequence of $q_i$
which is obtained from $q_i$ by removing all but one of consecutive items
of the same type for each type
and $C(Q)=\{C(q_1),\ldots,C(q_k)\}$.

\begin{observation}\label{main-con} 
Let $Q=\{q_1,\ldots,q_k\}$ be a set of $k$ sequences, 
then $\g(Q)=\g(C(Q))$.
\end{observation}

Next we give methods in order to compute the sequence digraph and also its complement digraph.
Therefore we define the position of the first item in some sequence $q_i\in Q$ of 
some type $t\in\types(Q)$  by
$\FI(q_i,t)$ and the position of the last item of type $t$ in
sequence $q_i$  by $\LA(q_i,t)$. For technical reasons, if 
there is no item for type $t$ contained in sequence $q_i$, then we define
$\FI(q_i,t) =n_i+1$, and $\LA(q_i,t) = 0$.

\begin{lemma}\label{fiandla-1}
Let  $Q=\{q_1\}$ be some set of one sequence, $\g(Q)=(V,A)$ the defined sequence digraph,
$\co(\g(Q))=(V,A^c)$ its complement digraph, and $u\neq v$ two vertices
of $V$.

\begin{enumerate}
\item \label{fiandla-1a}
There is an arc $(u,v)\in A$, if and only if 
$\FI(q_1,u)<\LA(q_1,v)$.


\item \label{fiandla-1b}
There is an arc $(u,v)\in A^c$, if and only if 
$\LA(q_1,v)<\FI(q_1,u)$.

\item \label{fiandla-1c}
If $(u,v)\in A^c$, then  $(v,u)\in A$.

\item \label{fiandla-1d}
There is an arc $(u,v)\in A$ and an arc $(v,u)\in A^c$, if and only if 
$\LA(q_1,u)<\FI(q_1,v)$.

\end{enumerate}
\end{lemma}

\sproof
\begin{enumerate}
\item By the definition of the arcs of sequence digraphs.

\item It holds $(u,v)\in A^c$ if and only if $(u,v)\not\in A$. Further 
by  (1) and $u\neq v$ it holds  $(u,v)\not\in A$ 
if and only if  $\FI(q_1,u)>\LA(q_1,v)$.

\item If $(u,v)\in A^c$, then  our result (\ref{fiandla-1b}) implies $\LA(q_1,v)<\FI(q_1,u)$
which implies that $(v,u)\in A$.

\item Follows by  (\ref{fiandla-1b}) and (\ref{fiandla-1c}).\eproof
\end{enumerate}

\begin{lemma}\label{fiandl-k}  
Let  $Q=\{q_1,\ldots,q_k\}$ be some set of $k$ sequences, $\g(Q)=(V,A)$ the defined sequence digraph, 
$\co(\g(Q))=(V,A^c)$ its complement digraph, 
and $u\neq v$ two vertices of $V$.
\begin{enumerate}
\item \label{fiandla-k-a}
There is an arc $(u,v)\in A$, if and only if there is some $q_i\in Q$ such that
$\FI(q_i,u)<\LA(q_i,v)$.

\item \label{fiandla-k-b}
There is an arc $(u,v)\in A^c$, if and only if for every $q_i\in Q$ it holds
$\LA(q_i,v)<\FI(q_i,u)$.
\end{enumerate}
\end{lemma}

\sproof
\begin{enumerate}
\item By the definition of the arcs of sequence digraphs.

\item It holds $(u,v)\in A^c$ if and only if $(u,v)\not\in A$. Further 
by (1) and $u\neq v$ we have  $(u,v)\not\in A$ 
if and only if for every $q_i\in Q$ it holds $\FI(q_i,u)>\LA(q_i,v)$.\eproof
\end{enumerate}

Properties (\ref{fiandla-1c}) and (\ref{fiandla-1d}) of Lemma \ref{fiandla-1}
can not be shown for $k>1$ sequences, since an arc $(u,v)\in A^c$ can be
obtained by two types $u$ and $v$ from two different sequences $q_i$ and $q_j$
such that $u\in\types(q_i)-\types(q_j)$ and  $v\in\types(q_j)-\types(q_i)$.
In such situations it is not necessary to have $(v,u)\in A$.

By Lemma \ref{fiandl-k}(\ref{fiandla-k-a})
only the first and the last item of each type in every $q_i\in Q$ are important
for the arcs in the  corresponding digraph. 
Let $M(q_i)$ be the subsequence of $q_i$
which is obtained from $q_i$ by removing all except the first and last item for each type
and $M(Q)=\{M(q_1),\ldots,M(q_k)\}$.

\begin{observation}\label{main_q} 
Let $Q=\{q_1,\ldots,q_k\}$ be a set of $k$ sequences, 
then $\g(Q)=\g(M(Q))$.
\end{observation}

%
%

Alternatively to Lemma \ref{fiandl-k}
the sequence digraph  can be obtained as follows.

\begin{proposition}\label{comp-sg}
Given a set $Q$ of sequences, the sequence digraph $\g(Q)=(V,A)$ can be computed in
time $\bigo(n + k\cdot |\types(Q)|^2)=\bigo(n + k\cdot |V|^2)$.
\end{proposition}

\begin{proof}
Digraph $\g(Q)=(V,A)$ can be computed in time                      
$\bigo(n + k\cdot |\types(Q)|^2)=\bigo(n + k\cdot |V|^2)$
by the algorithm {\sc Create Sequence Digraph} shown in
Figure \ref{fig:algorithm5}. A value is added to vertex set $V$ or arc set
$A$ only if it is not already contained. To check this efficiently in time
$\bigo(1)$ we have to implement $V$ and $A$ as boolean arrays.
Therefore we need some preprocessing phase where we run through
each sequence and seek for the types. This can be done
in time $\bigo(n+k\cdot |\types(Q)|)$. To implement
list $L$ efficiently, we have to use an additional boolean
array to test membership in $\bigo(1)$.
The inner loop sums up to $n = n_1 + \ldots + n_k$ steps.
Since $|A| \in \bigo(|V|^2)$, lines 6-8 will run in time $\bigo(|V|^2)$, 
so the overall running time is $\bigo(n + k \cdot |V|^2)$.
\end{proof}

\begin{figure}[htbp]
\hrule
{\strut\footnotesize \bf Algorithm {\sc Create Sequence Digraph}} 
\hrule
\begin{tabbing}
xxxx \= xxxx \= xxxx \= xxxx \= xxxx \= xxxx \= xxxx \=\kill
for $i := 1$ to $k$ do \\
\> add  $t(b_{i,1})$ to vertex set $V$, if it is not already contained \\
\> $L := (t(b_{i,1}))$ \>\>\>\>\>\> $\blacktriangleright$ $L$ contains
                   types of items up to item $b_{i,j}$ \\
\> for $j := 2$ to $n_i$ do \\
\>\> add $t(b_{i,j})$ to vertex set $V$, if it is not already contained \\
\>\> if ($j=\LA(q_i,t(b_{i,j}))$)\\
\>\>\> for each type $t' \in L$ do \\
\>\>\>\> if $t' \neq t(b_{i,j})$ add arc $(t',t(b_{i,j}))$ to  $A$, if it is
         not already contained \\
\>\> if ($j=\FI(q_i,t(b_{i,j}))$)\\
\>\>\>  add $t(b_{i,j})$ to list $L$
\end{tabbing}
\hrule
\caption{Create the sequence digraph $\g(Q)=(V,A)$ for some given set
  of sequences $Q$.}
\label{fig:algorithm5}
\end{figure}

\subsection{From Digraphs to Sequences}\label{sec-hardness}

\begin{definition}[Sequence System]
Let $G=(V,A)$ be some  digraph and $A=\{a_1,\ldots,a_\ell\}$ its arc set.
The {\em sequence system} $\q(G) = \{q_1,\ldots,q_\ell\}$ for $G$ is defined as follows.
\begin{enumerate}[(1)]
\item There are $2\ell$ items $b_{1,1},b_{1,2},\ldots,b_{\ell,1},b_{\ell,2}$.

\item Sequence $q_i=(b_{i,1},b_{i,2})$ for $1\leq i \leq \ell$.

\item The type of item $b_{i,1}$ is the first vertex of arc $a_i$
and the type of item $b_{i,2}$ is the second vertex of arc $a_i$ for
$1\leq i \leq \ell$. Thus  $\types(\q(G)) = V$.
\end{enumerate}
\end{definition}

\begin{example}[Sequence System]\label{EX3-new}
For the digraph $G$ of Figure \ref{F-ex-d}
the corresponding sequence system is given by $\q(G)=\{q_1,q_2,q_3,q_4,q_5,q_6,q_7\}$,
where $q_1 = [a,b]$,  $q_2 = [b,c]$,  $q_3 = [c,d]$, 
$q_4 = [d,e]$, $q_5 = [e,a]$, $q_6 = [e,f]$,  $q_7 = [f,a]$.
The sequence digraph of $\q(G)$ is $G$.
\end{example}

The definitions of the sequence system $\q(G)$ and the sequence digraph 
$\g(Q)$, defined in Section \ref{SCgb},  imply the following 
results.

\begin{observation}\label{prop}
For every digraph $G$ it holds $G = \g(\q(G))$.
\end{observation}

The related relation $Q=\q(\g(Q))$ is not true in general (e.g. not for $Q$ from
Example \ref{EX6}). But if the reduction to different consecutive items
$C(Q)$ contains exactly two items of different types 
the following equivalence holds true.


\begin{lemma}\label{le-q}
For every set $Q$ of sequences it holds $C(Q)=\q(\g(Q))$ 
if and only if  
each sequence  $q_i\in C(Q)$
contains exactly two items of different types.
\end{lemma}

Since $\q(\g(Q))=\q(\g(C(Q)))$ holds for every set of sequences $Q$, in Lemma \ref{le-q} we also 
can replace condition $C(Q)=\q(\g(Q))$  by condition $C(Q)=\q(\g(C(Q)))$.

\begin{lemma}\label{prop-b}
For every digraph $G=(V,A)$ with $\un(G)=(V,E)$ there is a set $Q$ of at most 
$|E|$ sequences such that $G=\g(Q)$.
\end{lemma}

\begin{proof}
For every digraph $G=(V,A)$ by Observation \ref{prop}
the sequence system $Q=\q(G)$ leads to a set of at most 
$|A|$ sequences such that $G=\g(Q)$. 
If we have two arcs $(u,v)$ and $(v,u)$ between two vertices $u$ and $v$ 
these can be represented by one sequence $[u,v,u]$.
\end{proof}

There are digraphs which even can be defined by one sequence 
(see Example \ref{le-cl}(\ref{le-cl-1})
and Theorem \ref{s11}) and there are digraphs for which  
$|E|$ sequences are really necessary (see Lemma \ref{le-pa} and Lemma \ref{c3}).

For digraphs of bounded vertex degree the sequence system $Q=\q(G)$ leads to sets
whose distribution and number of items of each type can be bounded as follows.

\begin{lemma}\label{lemma-c-da}
For every digraph $G=(V,A)$ where $\max(\Delta^-(G),\Delta^+(G))\leq d$ there is a set $Q$ 
with $d_Q\leq 2d$ and $c_Q\leq 2d$
such that $G=\g(Q)$.
\end{lemma}

In case of complete bioriented digraphs we can improve the
latter bounds as follows.

\begin{lemma}\label{lemma-c-d}
For every complete bioriented digraph $G=(V,A)$ where $\max(\Delta^-(G),\Delta^+(G))\leq d$ there is a set $Q$ 
with $d_Q\leq d$ and $c_Q\leq 2d$ (for $d\geq 2$ even $c_Q\leq 2d-1$)
such that $G=\g(Q)$.
\end{lemma}

\begin{proof}
By Lemma \ref{lemma-c-da} we obtain $d_Q\leq 2d$ and $c_Q\leq 2d$.
Since we consider complete bioriented digraphs for every two vertices $u$ and $v$
we now have two sequences $[u,v]$ and $[v,u]$, which can be 
combined to only one sequence $[u,v,u]$.
This leads to  $d_Q\leq d$ and $c_Q\leq 2d$. 
For $d\geq 2$ we can show the stated improvement as follows. 
Obviously for every weakly connected component of $G$
there is at most one type $t\in\types(Q)$ such that
there are exactly $2d$ items of type $t$. If there is one type $t\in\types(Q)$ such that
there are exactly $2d$ items of type $t$, then there are less than $d$
types $t'\in\types(Q)$ such that there are exactly $2d-1$ items of type $t'$.
Thus there is one sequence $[t,t'',t]$ such that there are
at most $2d-2$ items of type $t''$. If we substitute $[t,t'',t]$ by  $[t'',t,t'']$
we achieve $c_Q\leq2d-1$.
\end{proof}

\section{Properties of Sequence Digraphs}\label{properties}

\subsection{Graph Classes and their Relations}\label{classes}

In order to represent some digraph $G=(V,A)$ as a sequence digraph we need exactly
$|V|$ types.
By Lemma \ref{prop-b} every digraph is a sequence digraph using 
a suitable set $Q$ of sequences. Thus we want to consider  digraphs which can be
defined by a given upper bound for the number of sequences. 
Furthermore the first and the last item of each type within a sequence
is the most important one by Observation \ref{main_q}, thus we want to analyze the digraphs which can be
defined by a given upper bound for the number of items of each type. 
We define $S_{k,\ell}$
to be the set of all sequence digraphs defined by sets $Q$ 
on at most $k$ sequences that contain
at most $\ell$ 
items of each type in $\types(Q)$. By the definition we
know for every  two integers $k\geq 1$ and $\ell\geq 1$
the following inclusions between these graph classes.
\begin{eqnarray}
S_{k,\ell}    &\subseteq& S_{k+1,\ell}  \label{prop1a}\\
S_{k,\ell}    &\subseteq& S_{k,\ell+1}  \label{prop2b} 
\end{eqnarray}

By Lemma \ref{prop-b} and Observation \ref{main_q} we obtain the following bounds.

\begin{corollary}
Let $Q$ be a set on $k$ sequences and $\g(Q)=(V,A)\in S_{k,\ell}$  the defined graph
with $\un(\g(Q))=(V,E)$.
Then we can assume that $1\leq \ell\leq 2k$ and $1\leq k\leq |E|$.
\end{corollary}

\begin{lemma}\label{le-ud}
Let $\ell\geq 1$ and $G\in S_{1,\ell}$ be defined by $Q=\{q_1\}$, then digraph $\g(Q)$ is
semicomplete and thus graph $\un(\g(Q))$ is the complete graph on $|\types(Q)|$ 
vertices.
\end{lemma}

\begin{proof}
Let $Q=\{q_1\}$ be some set of one sequence which defines the sequence digraph $\g(Q)$.
For every vertex $v$ in $\g(Q)$ there is exactly one corresponding
type $t_{v}\in\types(q_1)$. Since $G\in S_{1,\ell}$ 
for every two vertices $v$ and $w$
in $\g(Q)$ there are two items $b_{1,i}$ and $b_{1,j}$ of type $t_v$ and 
$t_w$, respectively, which define the arc $(v,w)$ or $(w,v)$.
\end{proof}

\begin{corollary}\label{cc}
Let $G=(V,A)$ be a digraph, such that graph  $\un(G)$ has $k$ connected components. 
Then for every $k'<k$ and $\ell\geq 1$ it holds 
$G\not\in S_{k',\ell}$.
\end{corollary}

The following generalization of Lemma \ref{le-ud} to $k\geq 2$
sequences is easy to verify.

\begin{lemma}\label{l-con2}
Let $Q$ be some set of $k\geq 2$ sequences, then $\un(\g(Q))$ is connected if and only
if there is no set $Q'$, $\emptyset \neq Q'\subset Q$ such that  
$\types(Q')\cap \types(Q-Q')=\emptyset$.
\end{lemma}

Bounds on distribution and number of items for each type in $Q$ can
be used to classify $g(Q)$ into the classes $S_{k,\ell}$ as follows.

\begin{lemma}\label{l-c-d}
Let $Q$ be some set of $k$ sequences, then 
$\g(Q)\in S_{k, 2\cdot d_Q}\subseteq S_{k, 2\cdot k}$ 
and $\g(Q)\in S_{k, c_Q}$. 
\end{lemma}

\begin{proof}
Relation $\g(Q)\in S_{k, 2\cdot d_Q}$ holds by Observation \ref{main_q}. 
The further results hold by definition. 
\end{proof}

Next we consider the relations of the defined classes for
$k=1$ sequence.
Since $S_{1,1}$ contains only digraphs with exactly
one arc between every pair of vertices (cf. Theorem \ref{s11}
for a more precise characterization) and  $S_{1,\ell}$ for $\ell\geq 2$ 
contains all bidirectional complete
graphs we know that $S_{1,1}\subsetneq S_{1,\ell}$ for every $\ell\geq 2$.
Further by  (\ref{prop2b}) and  by Observation\ref{main_q}   
it follows that all classes $S_{1, \ell}$ for $\ell\geq 2$ are equal.

\begin{lemma}\label{le-1-leq}
For every two integers $\ell,\ell' \geq 2$ it holds $S_{1,\ell}=S_{1, \ell'}$.
\end{lemma}

\begin{corollary}\label{cor-1-seq} For $\ell \geq 2$ the following inclusions hold.
$$S_{1, 1}\subsetneq  S_{1, 2}= \ldots = S_{1, \ell}$$ 
\end{corollary}


The equalities of  Lemma \ref{le-1-leq}
can not be generalized for $k>1$ sequences and $\ell\leq k$ since removing items from
different sequences can change the sequence digraph. In order to
give examples for digraphs which do not belong to some of
the classes $S_{k,\ell}$ we next show some useful results.

For a set of digraphs ${\mathcal F}$ we denote by {\em ${\mathcal F}$-free digraphs} 
the set of all digraphs $G$ such that no induced subdigraph of $G$ is isomorphic
to a member of ${\mathcal F}$.
If  ${\mathcal F}$ consists of only one digraph $F$, we write $F$-free instead of $\{F\}$-free.
For undirected graphs we use this notation as well.

\begin{lemma}\label{c3}
Let $G=(V,E)$ be a triangle free graph, i.e. a $C_3$-free
graph,  with $|E|\geq 2$, such that  $\Delta(G)=\ell$ and $G'=(V,A)$ be an orientation
of $G$. 
Then for $k=|E|$ it holds $G'\in S_{k,\ell}$ but for $k'<k$ or $\ell'<\ell$ it holds 
$G'\not\in S_{k',\ell'}$.
\end{lemma}

\begin{proof}
Let $G$ and $G'$ be as in the statement of the lemma. Further let $Q$ be a set 
of sequences such that $G'=\g(Q)$. If some sequence  $q_i\in Q$ defines
more than one arc of $G'$, then $q_i$ has to contain at least three items
of pairwise distinct types. This would induce a $C_3$
in $G$, which is not possible by our assumption.
Thus  we have to represent every arc $(u,v)$ of $G'$
by one sequence $[u,v]$.
Thus $G'$ can not be defined by less than $k=|E|$ sequences or less than  $\ell=\Delta(G)$
items for every type.
\end{proof}

Since for every $k\geq 2$ and every  $\ell=2,\ldots,k$ 
there is a tree $T$ on $k$ edges and $\Delta(T)=\ell$ we 
know by Lemma \ref{c3} that for $k\geq 2$ and $\ell=2,\ldots,k$ we have
$S_{k,\ell-1} \subsetneq S_{k, \ell}$.
%
Further by Observation \ref{main_q}  we know that for $k\geq 2$ and $\ell\geq 2k$ it holds
$S_{k,\ell} = S_{k, \ell+1}$.

\begin{corollary}\label{cor2-seq} For  $k\geq 2$ 
the following inclusions hold.
$$S_{k, 1} \subsetneq S_{k, 2}  \subsetneq \ldots \subsetneq S_{k,k} \subseteq  S_{k,k+1}\subseteq  \ldots \subseteq  S_{k,2k} =  S_{k,2k+1} = \ldots$$
\end{corollary}

\begin{proposition}\label{th-ind-sg}
Let $G\in S_{k,\ell}$, then for every induced subdigraph 
$H$ of $G$ it holds  $H\in S_{k,\ell}$.
\end{proposition}

\begin{proof}
A set of sequences $Q'$ for an induced 
subdigraph $H$ can be obtained from a set $Q$ of the original digraph 
$G$ by restricting $Q'$ to the items which are destinated for types corresponding
to vertices of $H$.
\end{proof}

Graph classes which are closed under taking induced subgraphs
are called {\em hereditary}. Hereditary graph classes are exactly
those classes which can be defined by a set of forbidden induced
subgraphs. On the other hand, the classes $S_{k,\ell}$ are not closed 
under taking arbitrary subgraphs by the following example.

\begin{example}\label{le-cl}
\begin{enumerate}
\item \label{le-cl-1}
For every $n\geq 1$ and $\ell \geq 2$  it holds 
$\overleftrightarrow{K_n}\in S_{1, \ell}$, which can be
verified by set  $Q=\{q_1\}$, where
$$q_1=[v_1,v_2,\ldots,v_n,v_1,v_2,\ldots,v_n].$$

\item 
Every orientation $T'_{n,d}$ of a tree $T_{n,d}$ on $n$ vertices and $\Delta(T)=d\leq n-1$  can be obtained
as a subdigraph of $\overleftrightarrow{K_n}$. 
Since by Lemma \ref{c3}  for $n\geq 3$ it holds $T'_{n,d}\in S_{n-1,d}$ but $T'_{n,d}\not\in S_{n-2,d}$ 
and $T'_{n,d}\not\in S_{n-1,d-1}$
the classes  $S_{k, \ell}$ for each $k\geq 1$ and $\ell \geq 2$ 
are not closed under taking arbitrary subgraphs.

\item By Theorem \ref{s1k} digraph $TT_k$ which consists
of the disjoint union of $k$ transitive tournaments is 
in $S_{k,1}-S_{k-1,1}$.

\item If we remove exactly one edge from $TT_k$ we obtain
a digraph $TT'_k$ such that  $TT'_k\not\in S_{k,1}$. Thus 
the classes  $S_{k,1}$ for every $k\geq 1$ 
are not closed under taking arbitrary subgraphs.
\end{enumerate}
\end{example}

For orientations of paths and  cycles 
Lemma \ref{c3} leads to the following bounds.

\begin{lemma}\label{le-pa}
\begin{enumerate}
\item
Let $P'_n$ be an orientation of a path $P_n$ on $n\geq 2$ vertices, e.g. $P'_n=\overrightarrow{P_n}$
for $n\geq 2$. Then for every $\ell \geq 2$  it holds 
$P'_n\in S_{n-1, \ell}$.
But for every $\ell\geq 1$ it holds  $P'_n\not\in S_{n-2, \ell}$.

\item 
Let $C'_n$ be an orientation of a cycle $C_n$ on $n\geq 4$ vertices, e.g. $C'_n=\overrightarrow{C_n}$
for $n\geq 4$. Then for every $\ell \geq 2$  it holds 
$C'_n\in S_{n, \ell}$.
But for every $\ell\geq 1$ it holds  $C'_n\not\in S_{n-1, \ell}$.
\end{enumerate}
\end{lemma}

By Lemma \ref{le-pa} and Proposition \ref{th-ind-sg} we obtain the following results.

\begin{proposition}\label{le-pa2} Let $\ell \geq 1$, $k\geq 1$, and
$G\in S_{k,\ell}$, then $G$ has no induced path $\overrightarrow{P_{k'}}$ 
for $k'\geq k+2$. 
\end{proposition}

\begin{proposition}\label{le-pa3}  
Let $\ell \geq 1$, $k\geq 3$, and $G\in S_{k,\ell }$, then $G$ has no induced cycle 
$\overrightarrow{C_{k'}}$ for $k'\geq k+1$.
\end{proposition}

\subsection{Characterizations of Sequence Digraphs for $k=1$ or $\ell=1$}\label{sec-char-1s}

By Proposition \ref{th-ind-sg} all the classes $S_{k,\ell}$ are  hereditary 
and thus can be characterized by a set of forbidden induced subgraphs. 
In this section we even show a finite set of  forbidden induced subgraphs
for all classes $S_{k,\ell}$ where $\ell=1$ and for all classes where $k=1$. 
These characterizations lead to polynomial time  recognition algorithms for 
the corresponding graph classes. Furthermore we give characterizations
in terms of special tournaments and conditions for the complement digraph.

\subsubsection{Sequence Digraphs for $k=1$ and $\ell=1$}\label{sec-1-1}

Regarding the algorithmic use of special digraphs it is often helpful to know
whether they are transitive or at least quasi transitive.
A digraph $G=(V,A)$ is called {\em transitive} if for 
every pair $(u,v)\in A$ and $(v,w)\in A$ of arcs
with $u\neq w$ the arc $(u,w)$ also belongs to $A$.

\begin{lemma}\label{l1}
Every digraph in $S_{1,1}$ is transitive.
\end{lemma}

\begin{proof}
Let $Q=\{q_1\}$ be some set of one sequence. 
If for every type there is only one item in $q_1$, then two arcs
$(u,v)$ and $(v,w)$ can only be defined in $\g(Q)$ if $q_1$ contains some item $b_{1,i}$
of type $u$, some item $b_{1,i'}$
of type $v$, and some item $b_{1,i''}$
of type $w$ such that $i<i'$ and $i'<i''$. By the definition of the
sequence digraph $\g(Q)$ also has the arc $(u,w)$.
\end{proof}

In order to characterize digraphs in $S_{1,1}$ we recall 
the following operations which are used in
the definition of oriented threshold graphs in \cite{Boe15}. For 
some digraph $G=(V,A)$ and some operation
$o\in\{\text{IV},\text{OD},\text{ID}\}$ we define by $G'=o(G)$ the digraph 
with vertex set $V\cup \{v\}$ such that $G'[V]=G$ and
\begin{itemize}
\item $N^+_{G'}(v)=\emptyset$ and  $N^-_{G'}(v)=\emptyset$ ($v$ is an isolated vertex, IV for short)
\item $N^+_{G'}(v)=V$ and $N^-_{G'}(v)=\emptyset$   ($v$ is an out-dominating vertex, OD for short)
\item $N^+_{G'}(v)=\emptyset$ and $N^-_{G'}(v)=V$ ($v$ is an in-dominated vertex, ID for short)
\end{itemize}
For some set of operations $\mathcal{O}\subseteq \{\text{IV},\text{OD},\text{ID}\}$ 
let $\mathcal{G_O}$ be the following set of all digraphs.
\begin{itemize}
\item The one vertex graph $(\{v\},\emptyset)$ is in $\mathcal{G_O}$ and
\item for every digraph $G\in \mathcal{G_O}$ and operation $o\in \mathcal{O}$ the 
digraph $o(G)$ is in $\mathcal{G_O}$.
\end{itemize}

For some digraph $G$ and some integer $d$ let $d G$ be the disjoint union
of $k$ copies of $G$.

\begin{theorem}\label{s11}
For every digraph $G$ the following statements are equivalent.
\begin{enumerate}
\item \label{s11a} $G\in S_{1,1}$ 
\item \label{s11b} $G$ is a transitive tournament.
\item \label{s11c} $G$ is an acyclic tournament.
\item \label{s11d} $G$ is a $\overrightarrow{C_3}$-free tournament.
\item \label{s11e} $G$ is a tournament with exactly one Hamiltonian path.
\item \label{s11f} $G$ is a tournament and every vertex in $G$ 
has a different outdegree, i.e. $\{\odeg(v)~|~v\in V\}=\{0,\ldots,|V|-1\}$.
\item \label{s11g} $G$ is $\{2\overleftrightarrow{K_1},\overleftrightarrow{K_2},\overrightarrow{C_3}\}$-free.  
\item \label{s11po} $G\in \{(\{v\},\emptyset)\}\cup\{(\overrightarrow{P_n})^{n-1}~|~n\geq 2\}$, i.e. $G$
is the $(n-1)$-th power of a directed path $\overrightarrow{P_n}$.
\item \label{s11h} $G\in \mathcal{G}_{\{\text{OD}\}}$, i.e. $G$ can 
be constructed from the one-vertex graph $K_1$ by repeatedly adding
an  out-dominating vertex.
\item \label{s11i} $G \in \mathcal{G}_{\{\text{ID}\}}$, i.e. $G$ 
can be constructed from the one-vertex graph $K_1$ by repeatedly adding
an   in-dominated  vertex.
\end{enumerate}
\end{theorem}

\begin{proof}
The equivalence of $(\ref{s11b})-(\ref{s11f})$  is  known from \cite[Chapter 9]{Gou12}.

$(\ref{s11a})\Rightarrow (\ref{s11b})$ By Lemma \ref{l1} digraph in $S_{1,1}$  is transitive and 
by definition of $S_{1,1}$ 
digraph $G$ is a  tournament.

$(\ref{s11c})\Rightarrow (\ref{s11a})$ Every acyclic digraph $G$ 
has a source, i.e.~a vertex $v_1$ of indegree $0$, see \cite{BG09}. Since $G$ is a tournament 
there is an arc $(v_1,v)$ for every vertex $v$ of $G$, i.e. $v_1$ is an 
out-dominating vertex. By removing $v_1$ from $G$, we
obtain a transitive tournament $G^1$ which leads to an out-dominating vertex $v_2$. By removing 
$v_2$ from $G^1$, we
obtain a transitive tournament $G^2$ which leads to  an out-dominating vertex $v_3$ 
and so on. The sequence
$[v_1,v_2,\ldots,v_n]$ shows that $G\in S_{1,1}$.

$(\ref{s11d})\Leftrightarrow (\ref{s11g})$ Obviously

$(\ref{s11a})\Leftrightarrow (\ref{s11h})$, $(\ref{s11a})\Leftrightarrow (\ref{s11i})$, and $(\ref{s11a})\Leftrightarrow (\ref{s11po})$ hold by 
definition.
\end{proof}

\begin{proposition}\label{le-find11}
Let $G=(V,A)\in S_{1,1}$, then a set $Q$ of one sequence $q$, such
that $G=\g(Q)$ can be found in time $\bigo(|V|+|A|)$.
\end{proposition}

\begin{proof}
A  method is given in part $(\ref{s11c})\Rightarrow (\ref{s11a})$
of the proof of Theorem \ref{s11}.
\end{proof}

\subsubsection{Sequence Digraphs for $\ell=1$}

The sequence digraph $\g(Q)=(V,A)$ for a set $Q=\{q_1, \ldots ,q_k\}$
can be obtained by the union of $\g(\{q_i\})=(V_i,A_i)$, $1\leq i\leq k$ by
$V=\cup_{i=1}^k V_i$ and $A=\cup_{i=1}^k A_i$.
Since for digraphs in $S_{k,1}$ the vertex sets $V_i=\types(q_i)$
are disjoint, we can follow the next lemma. 

\begin{lemma}\label{unionsk1}
For every integer $k$ every digraph in $S_{k,1}$ is the disjoint union of $k$
digraphs in $S_{1,1}$. 
\end{lemma}

This lemma allows us to generalize Lemma \ref{l1} to 
$k$ sequences.

\begin{corollary}
For every integer $k$ every digraph in $S_{k,1}$ is transitive.
\end{corollary}

On the other hand none of the sets  $S_{k,1}$ contains all
transitive digraphs. Therefore let
$T'$ be an orientation of a tree $T$ with $k+1$ edges, 
such that for every vertex $v$ in
$T'$ either the indegree or outdegree is $0$. Such an
orientation is possible for every tree $T$ and leads to a transitive 
graph $T'$. By Lemma \ref{c3} we know that $T'\not\in S_{k,1}$.

Theorem \ref{s11} can be generalized to $k\geq 1$ sequences as follows.

\begin{theorem}\label{s1k}
For every digraph $G$ and every integer $k\geq 1$ the following statements are equivalent.
\begin{enumerate}
\item \label{s1ka} $G\in S_{k,1}$.
\item \label{s1kaa} $G$ is the disjoint union of $k$  digraphs from $S_{1,1}$.
\item \label{s1kb} $G$ is the disjoint union of $k$  transitive tournaments.
\item \label{s1kc} $G$ is the disjoint union of $k$  acyclic tournaments.
\item \label{s1kd} $G$ is the disjoint union of $k$ $\overrightarrow{C_3}$-free tournaments.
\item \label{s1kg} $G$ is $\{(k+1)\overleftrightarrow{K_1},\overleftrightarrow{K_2},\overrightarrow{C_3}\}$-free.  
\end{enumerate}
\end{theorem}

\begin{proposition}\label{le-find11a}
Let $G=(V,A)\in S_{k,1}$, then a set $Q$ on $k$ sequences, such
that $G=\g(Q)$ can be found in time $\bigo(|V|+|A|)$.
\end{proposition}

\begin{proof}
If  $G\in S_{k,1}$ then $G$ is the  disjoint union of $k$  digraphs from $S_{1,1}$.
For every of them we can define a sequence $q_i$ by Proposition \ref{le-find11}.
\end{proof}

\subsubsection{Sequence Digraphs for $k=1$}

The next examples show that for $\ell\geq 2$ 
items for each type even one sequence can define
digraphs which are not transitive.

\begin{example}\label{notta}
\begin{enumerate}
\item \label{notta1}
The digraph $D_0$ in Table \ref{F-co2} is not transitive, since
it has among others
the arcs $(b,c)$ and $(c,a)$ but not the arc $(b,a)$.
Further $D_0$ belongs to the set $S_{1,2}$, since
it can be defined by set $Q = \{q_1\}$ of one sequence
$q_1=[c,a,b,c]$.

\item \label{notta2}
The digraph $D_6$ in Table \ref{F-co2} (which will be of further
interest in Section \ref{sec-co}) is not transitive, since
it has among others
the arcs $(c,b)$ and $(b,a)$ but not the arc $(c,a)$.
Further $D_6$ belongs to the set $S_{1,2}$, since
it can be defined by set $Q = \{q_1\}$ of one sequence
$q_1=[a,b,a,c,b,d,c,d]$.
\end{enumerate}
\end{example}

\begin{table}
\begin{center}
\begin{tabular}{cccccccc}

\epsfig{figure=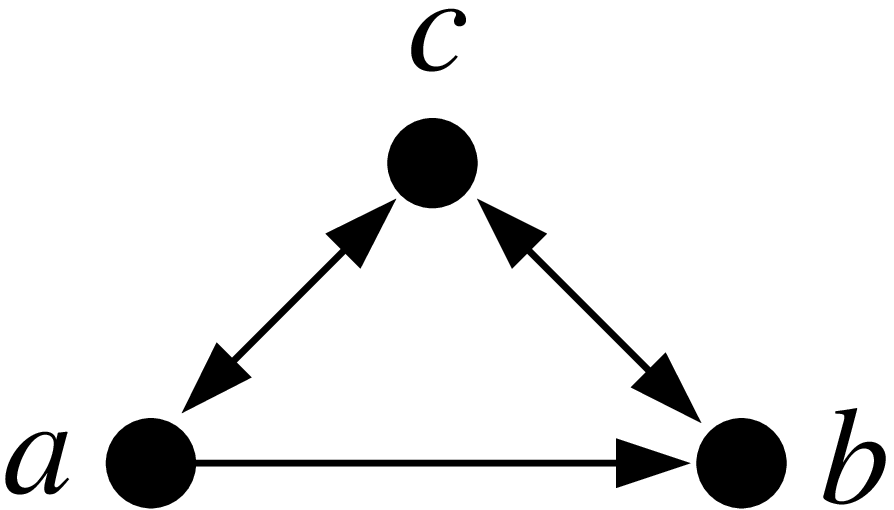,width=2.5cm} &&\epsfig{figure=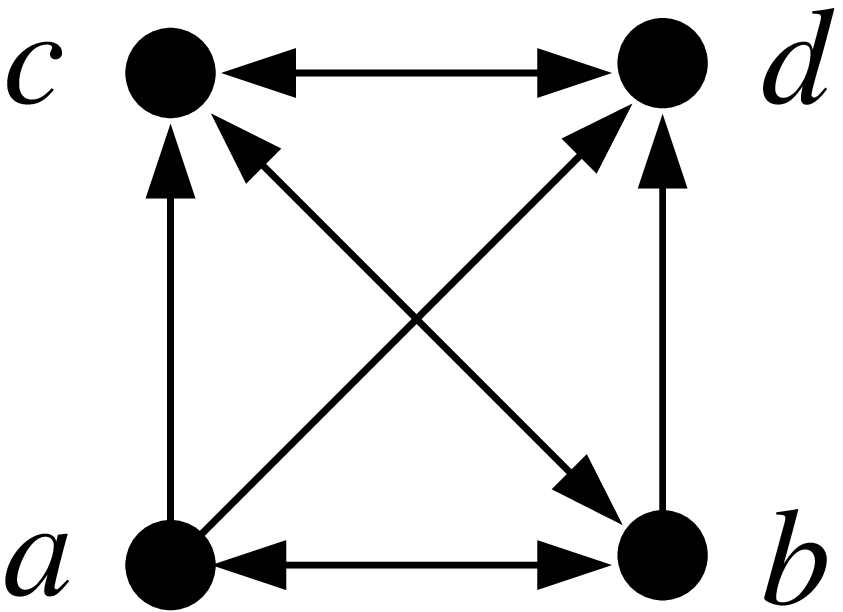,width=2.4cm} &&\epsfig{figure=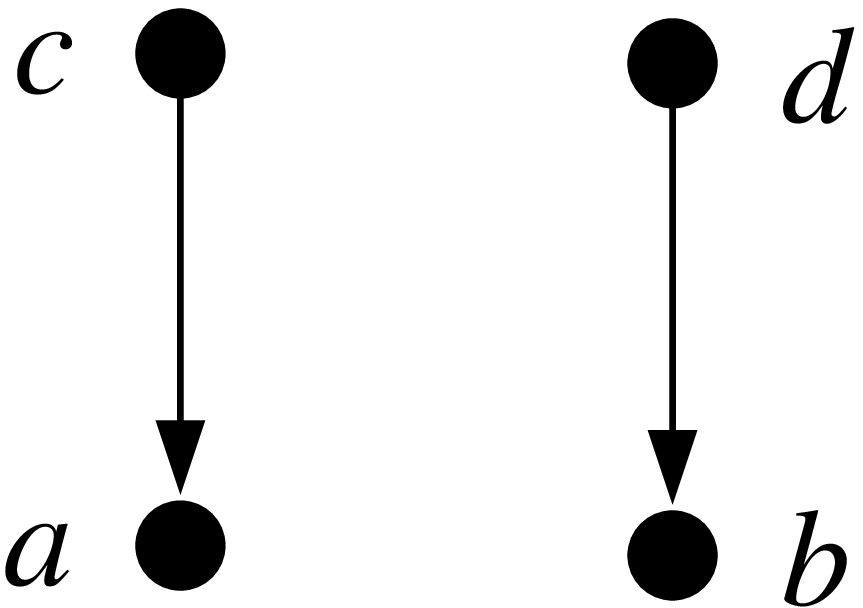,width=2.4cm} &&\epsfig{figure=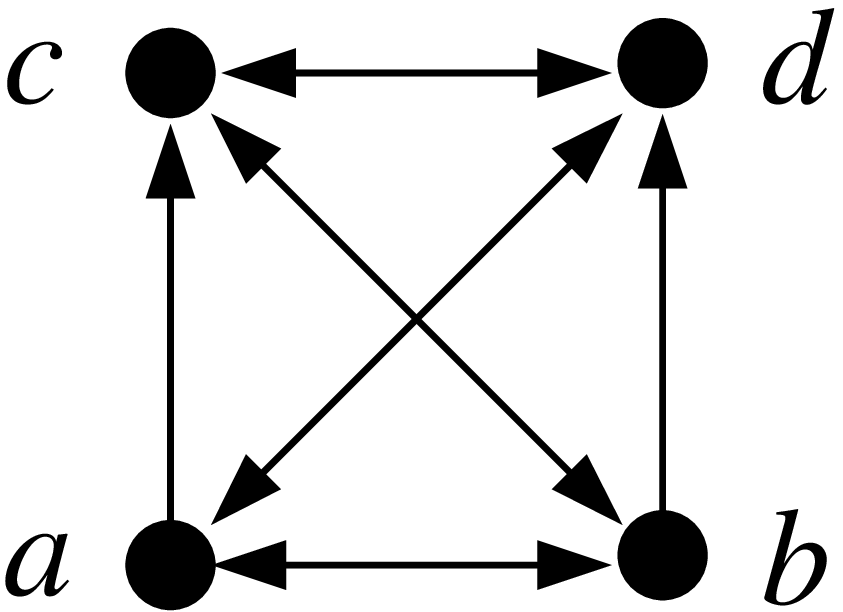,width=2.4cm}\\
$D_0$                              &&$D_6$                                     && $2\overrightarrow{P_2}$            &&  $\co(2\overrightarrow{P_2})$      \\ 
\end{tabular}
\end{center}
\caption{Special digraphs.}
\label{F-co2}
\end{table}

A digraph $G=(V,A)$ is called {\em quasi transitive} if for 
every pair $(u,v)\in A$ and $(v,w)\in A$ of arcs 
with $u\neq w$ there is at least one arc between $u$ and $w$ in $A$.
Since every semicomplete digraph is quasi transitive, 
by Lemma \ref{le-ud} the next result holds true.

\begin{lemma}\label{l2}
For every integer  $\ell\geq 1$ every digraph in $S_{1,\ell}$ is quasi transitive.
\end{lemma}

On the other hand none of the sets $S_{1,\ell}$ contains all
quasi transitive digraphs. Therefore let $G$ be a tournament
without any out-dominating (or without any in-dominated) vertex.
Every such digraph $G$ is quasi transitive.
Then $G\not\in S_{1,1}$ by Theorem \ref{s11} (\ref{s11h}.) (or (\ref{s11i}.)) and
$G\not\in S_{1,2}-S_{1,1}$ since every digraph
of this set has a cycle on two vertices. Since $S_{1,2}= S_{1,\ell}$
for $\ell\geq 3$, see Lemma \ref{le-1-leq}, we know that $G\not\in S_{1,\ell}-S_{1,1}$
for $\ell \geq 2$.

Although we want to characterize graphs in $S_{k,\ell}$ for $k=1$
we give a small remark for $k\geq 1$.

\begin{remark}\label{cor-union-transi}
For every $\ell\geq 1$ and every $k\geq 1$ every digraph in $S_{k,\ell}$ 
is the union of at most $k$ quasi 
transitive digraphs. There are even digraphs in  $S_{k,2}$ for which we need the union of $k$ quasi 
transitive digraphs, e.g. the directed path on $n$ vertices
$\overrightarrow{P_n}$ can only be obtained by the union of $n-1$ quasi 
transitive digraphs each defining one arc, see Lemma \ref{le-pa}.
\end{remark}

In order to show a characterization for the class $S_{1,2}$
we next prove lemmas.

\begin{lemma}\label{le-transi}
Let $\ell\geq 1$ and $G\in S_{1,\ell}$, then the 
complement digraph $\co G$ is transitive.
\end{lemma}

\begin{proof}
Let $Q=\{q_1\}$ be some sequence which defines the sequence digraph $\g(Q)$ and 
$(u,v)$ and $(v,w)$ with $u\neq w$ be two arcs of $\co(\g(Q))$.
By
Lemma \ref{fiandla-1} (\ref{fiandla-1b}) we conclude that
$\LA(q_1,v)<\FI(q_1,u)$
and 
$\LA(q_1,w)<\FI(q_1,v)$
thus we have
$\LA(q_1,w)<\FI(q_1,v) \leq \LA(q_1,v)<\FI(q_1,u)$
and $\LA(q_1,w)<\FI(q_1,u)$ implies by 
Lemma \ref{fiandla-1} (\ref{fiandla-1b}) 
that arc $(u,w)$ is in digraph $\co(\g(Q))$.
\end{proof}

\begin{lemma}\label{le-p2free}
Let $\ell\geq 1$ and $G\in S_{1,\ell}$, then 
the complement digraph $\co G$ is
$2\overrightarrow{P_2}$-free.
\end{lemma}

\begin{proof}
If we assume that $\co G$  contains 
$2\overrightarrow{P_2} =(\{a,b,c,d\},\{(c,a),(d,b)\})$, see
Table \ref{F-co2}, as an induced subdigraph, then $G$
has (among others) the arcs 
$(c,b)$ and $(d,a)$.
By applying Lemma \ref{fiandla-1}(\ref{fiandla-1a}) and (\ref{fiandla-1b}) we obtain
$\LA(q_1,a)<\FI(q_1,c)<\LA(q_1,b)<\FI(q_1,d)<\LA(q_1,a)$,
which leads to a contradiction.
\end{proof}

\begin{lemma}\label{le-ind-out-d}
Let $G$ be a semicomplete   $\{\overrightarrow{C_3},D_4\}$-free digraph on $n$
vertices, then $G$ has a vertex $v$ 
such that $\odeg(v)=n-1$ and a vertex $v'$ such that $\ideg(v')=n-1$. ($D_4$ is shown in Table \ref{F-co}.)
\end{lemma}

\begin{proof}
We show the result by induction on the number of vertices $n$ in digraph $G$. 
For digraphs on $n\leq 3$ vertices the statement is true.\footnote{A list of
all digraphs on at most four vertices can be found in Appendix 2 of \cite{Har69}.}

Let $n\geq 3$ and $G_{n+1}=(V_{n+1},A_{n+1})$ be a  semicomplete   $\{\overrightarrow{C_3},D_4\}$-free 
digraph on $n+1$ vertices. Further let $w\in V_{n+1}$ and $G_n=G_{n+1}[V_{n+1}-\{w\}]$. 
Then $G_n=(V_n,A_n)$ is a  semicomplete $\{\overrightarrow{C_3},D_4\}$-free 
digraph on $n$ vertices. By  the induction hypothesis there is some $v\in V_n$ 
such that $\odeg(v)=n-1$ in  $G_n$.

If $(v,w)\in A_{n+1}$ then  $\odeg(v)=n$ in $G_{n+1}$ and the proof is done.
If $(v,w)\not\in A_{n+1}$ then the semicompleteness of $G_{n+1}$  implies that 
$(w,v)\in A_{n+1}$.  Then for every $(v,u)\in E_n$ it holds $(w,u)\in E_{n+1}$.
Otherwise  the semicompleteness of $G_{n+1}$  implies that 
$(u,w)\in E_{n+1}$ and thus $G_{n+1}[\{u,v,w\}]$ leads to a $\overrightarrow{C_3}$ or to a $D_4$.
Thus in $G_{n+1}$ it holds  $\odeg(w)=\odeg(v)+1=n-1+1=n$. 

The proof for the existence of a vertex  $v'$ such that $\ideg(v')=n-1$ is similar.
\end{proof}

\begin{lemma}\label{le-co-t}
Let $G$ be a semicomplete  $\{\overrightarrow{C_3},D_4\}$-free digraph, then
its complement digraph $\co G$ is transitive. ($D_4$ is shown in Table \ref{F-co}.)
\end{lemma}

\begin{proof}
Let  $(u,v),(v,w)\in A^c$ be two arcs of  $\co G=(V,A^c)$.
Since $G=(V,A)$ is semicomplete we know that $(v,u),(w,v)\not\in A^c$.
If $G$ is a $\{\overrightarrow{C_3},D_4\}$-free digraph, then $\co G$ 
is a $\{\overrightarrow{C_3},\overrightarrow{P_3}\}$-free digraph. 
Thus $u$ and $w$ are connected either only by
$(u,w)\in A^c$ or by $(u,w)\in A^c$ and $(w,u)\in A^c$, which
implies that $\co G$ is transitive.
\end{proof}

A {\em spanning} subdigraph is a subdigraph obtained by the deletion of arcs.

\begin{lemma}\label{le-sp-tt}
Every semicomplete $\{\overrightarrow{C_3},D_4\}$-free digraph
has a spanning transitive tournament subdigraph. ($D_4$ is shown in Table \ref{F-co}.)
\end{lemma}

\begin{proof}
By Lemma \ref{le-ind-out-d} every  semicomplete $\{\overrightarrow{C_3},D_4\}$-free digraph $G$
on $n$ vertices has a
vertex $v_1$  such that $\odeg(v_1)=n-1$. By removing $v_1$ from $G$, we obtain 
a semicomplete $\{\overrightarrow{C_3},D_4\}$-free digraph $G^1$
on $n-1$ vertices, which leads to a vertex $v_2$  such that $\odeg(v_2)=n-2$. 
By removing $v_2$ from $G^1$, we obtain a semicomplete $\{\overrightarrow{C_3},D_4\}$-free 
digraph $G^2$ on $n-2$ vertices, which leads to a vertex $v_3$  such that $\odeg(v_3)=n-3$ and so on.
The sequence
$q=[v_1,v_2,\ldots,v_n]$ defines a digraph $\g(\{q\})\in S_{1,1}$ which is a spanning 
transitive tournament subdigraph
of $G$.
\end{proof}

These results allow us to show the following characterizations
for the class $S_{1,2}$. Since we use several forbidden induced subdigraphs
the  semicompleteness is expressed by excluding 
$2\overleftrightarrow{K_1}$.

\begin{theorem}\label{s12}
For every digraph $G$ the following statements are equivalent.
\begin{enumerate}
\item \label{s12-a} $G\in S_{1,2}$ 
\item \label{s12-aa} $G\in S_{1,\ell}$ for some $\ell\geq 2$ 
\item \label{s12-b} $\co G$ is transitive, $\co G$ is
$2\overrightarrow{P_2}$-free, and  
$G$ has a spanning transitive tournament subdigraph. 
\item  \label{s12-c} $G$ is $\{\co(2\overrightarrow{P_2}),2\overleftrightarrow{K_1},\overrightarrow{C_3},D_4\}$-free. 
($D_4$ is shown in Table \ref{F-co}.)
\end{enumerate}
\end{theorem}

\begin{proof}
$(\ref{s12-a})\Rightarrow (\ref{s12-aa})$ By definition. 
$(\ref{s12-aa})\Rightarrow (\ref{s12-a})$ 
By Lemma \ref{le-1-leq}.

$(\ref{s12-b})\Rightarrow (\ref{s12-c})$ If $G$ has a  spanning transitive tournament subdigraph
then it is semicomplete and thus $2\overleftrightarrow{K_1}$-free. If $\co G$ is transitive
it has no induced $\overrightarrow{C_3}$ and no $\overrightarrow{P_3}$. Thus $G$ has
no induced  $\overrightarrow{C_3}$ and $D_4$.
$(\ref{s12-c})\Rightarrow (\ref{s12-b})$ Follows by Lemma \ref{le-co-t} and Lemma \ref{le-sp-tt}.

$(\ref{s12-a})\Rightarrow (\ref{s12-b})$ Digraph $\co G$ is transitive
by Lemma \ref{le-transi} and $2\overrightarrow{P_2}$-free by Lemma \ref{le-p2free}. 
Further if $G$ is 
defined by sequence $q$ then the subsequence $F(q)$ 
which is obtained from $q$ by removing all except the first item for each type
leads to a subdigraph $(V,A')$ which is a transitive tournament.

$(\ref{s12-b})\Rightarrow (\ref{s12-aa})$ Let $G'=(V,A')$ be a subdigraph of $G=(V,A)$ 
which is a transitive tournament. 
By Theorem \ref{s11} we know that $G'\in S_{1,1}$ and thus 
there is some sequence 
$$q'=[v_1,\ldots,v_n]$$ such that $\g(\{q'\})=G'$. 
If $A'=A$ we know that $G\in S_{1,1}\subseteq S_{1,\ell}$ for every $\ell\geq 2$. 
So we can assume that $A'\subsetneq A$. Obviously for every 
arc  $(v_i,v_j)\in A-A'$  there are two positions $j<i$  in
$$q'=[v_1,\ldots,v_j,\ldots, v_i,\ldots, v_n].$$

In order to define a subdigraph of $G$ which contains
all arcs of $G'$ and arc $(v_i,v_j)$ we can 
insert (cf. Section \ref{SCgb} for the definition of inserting an item) an additional item for type $v_i$ 
on position $k\leq j$, or  an additional item for type $v_j$ 
on position $k> i$, or first an additional item 
for type $v_j$ and then an additional item for type $v_i$  
on a position $k$, $j<k\leq i$, into $q'$ without creating 
an arc which is not in $A$.
This is  possible if and only if there is some position $k$, $j\leq k \leq i$, in 
$$q'=[v_1,\ldots,v_j,\ldots, v_{m'}, \ldots, v_k, \ldots, v_{m''}, \ldots, v_i, \ldots,  v_n]$$
such that for every $m'$, $j< m ' \leq k$, it holds $(v_{m'},v_j)\in A$
and for every $m''$, $k\leq m''< i$, it holds $(v_i,v_{m''})\in A$.

If it is possible to insert all arcs of $A-A'$ 
by adding a set of additional items into sequence $q'$ resulting in a sequence $q$ 
such that $G=\g(q)$, then we have $G\in S_{1,\ell}$ for some $\ell\geq 2$.
%
Next we show a condition using the new items of every single arc  of $A-A'$  independently
from each other. 

\begin{claim}\label{cl} If for every arc  $(v_i,v_j)\in A-A'$ 
there is  a position $k$, $j<k\leq i$ such that first inserting 
an additional item for type $v_j$ and then an additional item for type $v_i$ 
at position $k$ into $q'$ defines a subdigraph of $G$ which contains
all arcs of $G'$ and arc $(v_i,v_j)$, then $G\in S_{1,\ell}$ for some $\ell\geq 2$.
\end{claim}

\begin{proof}{(\em of Claim \ref{cl})}
Every single arc  $(v_i,v_j)\in A-A'$  can be inserted by 
first inserting first an additional item for type $v_j$ and then 
an additional item for type $v_i$ at some  position $k$, $j<k\leq i$. 
Next we show that the new inserted 
items for two missing arcs  $(v_i,v_j)\in A-A'$ at position $k$ and $(v_{i'},v_{j'})\in A-A'$
at position $k'$ do not create an arc which is not in $A$. This is done by a case
distinction w.r.t.~the 34 possible positions of $i,i',j,j',k,k'$
shown in  Table \ref{cases-tab}.

\begin{enumerate}
\item If $j'<k'<i'<j<k<i$, then $(v_{j'},v_i)\in A$, $(v_{i'},v_i)\in A$, $(v_{j'},v_j)\in A$, and $(v_{i'},v_j)\in A$.
\item If $j'<k'<i'=j<k<i$, then $(v_{j'},v_i)\in A$, $(v_{i'},v_i)\in A$, and $(v_{j'},v_j)\in A$.
\item If $j'<k'<j<i'<k<i$, then $(v_{j'},v_i)\in A$, $(v_{i'},v_i)\in A$, $(v_{j'},v_j)\in A$, and $(v_{i'},v_j)\in A$.
\item If $j'<k'<j<k<i'<i$, then $(v_{j'},v_i)\in A$, $(v_{i'},v_i)\in A$, $(v_{j'},v_j)\in A$, and $(v_{i'},v_j)\in A$.
\item If $j'<k'<j<k<i'=i$, then $(v_{j'},v_i)\in A$, $(v_{i'},v_i)\in A$, $(v_{j'},v_j)\in A$, and $(v_{i'},v_j)\in A$.
\item If $j'<k'<j<k<i<i'$, then $(v_{j'},v_i)\in A$, $(v_{i'},v_i)\in A$, $(v_{j'},v_j)\in A$, and $(v_{i'},v_j)\in A$.
\item If $j'<j<k'<i'<k<i$, then $(v_{j'},v_i)\in A$, $(v_{i'},v_i)\in A$, $(v_{j'},v_j)\in A$, and $(v_{i'},v_j)\in A$.
\item\label{i-8} If $j'<j<k'<k<i'<i$, then $(v_{j'},v_i)\in A$, $(v_{i'},v_i)\in A$, $(v_{j'},v_j)\in A$. 
But arc $(v_{i'},v_j)$ does not belong to $A$ because of
the positions of $i,i',j,j'$. Next we show that $(v_{i'},v_j)\in A$ using  $(v_{i},v_{j'})$
which also does not belong to $A$ because of
the positions of $i,i',j,j'$.
If $(v_{i'},v_j)\not\in A$ and $(v_{i},v_{j'})\not\in A$ then  $j'<j<k'<k<i'<i$ implies that
$\co G$ has an induced $2\overrightarrow{P_2}$ which is not possible by our assumption.
Thus we know that $(v_{i'},v_j)\in A$ or $(v_{i},v_{j'})\in A$.
If $(v_{i},v_{j'})\in A$ we can exchange the positions of $k'$ and $k$ (which implies
that we are in case \ref{i-11}) or we also have $(v_{i'},v_j)\in A$.

\item If $j'<j<k'<k<i'=i$, then  $(v_{j'},v_i)\in A$, $(v_{j'},v_j)\in A$, and $(v_{i'},v_j)=(v_i,v_j)\in A$.

\item If $j'<j<k'<k<i<i'$, then  $(v_{j'},v_i)\in A$, $(v_{i'},v_i)\in A$, $(v_{j'},v_j)\in A$, and $(v_{i'},v_j)\in A$ 
(same situation as \ref{i-8}).

\item\label{i-11} If $j'<j<k<k'<i'<i$, then  $(v_{j},v_{i'})\in A$, $(v_{i},v_{i'})\in A$, $(v_{j},v_{j'})\in A$.
But arc $(v_{i},v_{j'})$ does not belong to $A$ because of
the positions of $i,i',j,j'$. Next we show that $(v_{i},v_{j'})\in A$ using  $(v_{i'},v_{j})$
which also does not belong to $A$ because of
the positions of $i,i',j,j'$.
If $(v_{i},v_{j'})\not\in A$ and $(v_{i'},v_{j})\not\in A$ then $j'<j<k<k'<i'<i$  implies that
$\co G$ has an induced $2\overrightarrow{P_2}$ which is not possible by our assumption.
Thus we know that $(v_{i},v_{j'})\in A$ or $(v_{i'},v_{j})\in A$.
If $(v_{i'},v_{j})\in A$ we can exchange the positions of $k'$ and $k$ (which implies
that we are in case \ref{i-8}) or we also have $(v_{i},v_{j'})\in A$.

\item If $j'<j<k<k'<i'=i$, then  $(v_{j},v_{i'})\in A$, $(v_{j},v_{j'})\in A$, and $(v_{i},v_{j'})=(v_{i'},v_{j'})\in A$.
\item If $j'<j<k<k'<i<i'$, then  $(v_{j},v_{i'})\in A$, $(v_{i},v_{i'})\in A$, $(v_{j},v_{j'})\in A$, and $(v_{i},v_{j'})\in A$ (same situation as \ref{i-11}).
\item If $j'<j<k<i<k'<i'$, then  $(v_{j},v_{i'})\in A$, $(v_{i},v_{i'})\in A$, $(v_{j},v_{j'})\in A$, and $(v_{i},v_{j'})\in A$.
\item If $j'=j<k'<i'<k<i$, then $(v_{j'},v_i)\in A$, $(v_{i'},v_i)\in A$, and $(v_{i'},v_j)\in A$.
\item If $j'=j<k'<k<i'<i$, then $(v_{j'},v_i)\in A$, $(v_{i'},v_i)\in A$, and $(v_{i'},v_j)=(v_{i'},v_{j'})\in A$.
\item If $j'=j<k'<k<i<i'$, then $(v_{j'},v_i)\in A$, $(v_{i'},v_i)\in A$, and $(v_{i'},v_j)=(v_{i'},v_{j'})\in A$.
\end{enumerate}

Every remaining case $i\in\{18,\ldots,34\}$ is symmetric to shown case $35-i$.
Thus for every two missing arcs  $(v_i,v_j)\in A-A'$ and $(v_{i'},v_{j'})\in A-A'$
which are inserted by additionally items all arcs between these items belong to $A$.

\medskip
Now we can show 
the claim by induction on the number $m=|A-A'|$ of missing arcs. Let $a_1,\ldots,a_m$
be an arbitrary order of the missing arcs.
For $m=1$ we can insert the missing arc by the assumption of the claim.
Let $m\geq 1$, $q_m$ be the sequence obtained from $q'$ by inserting $m$ 
missing arcs, and $G_m=\g(q_m))=(V,A_m)$. In order to insert arc  $a_{m+1}=(v_i,v_j)$ we
neglect the $2m$ new items in $q_m$ and  consider only the items of $q'$
for determining a position $k_{a_{m+1}}$, $j<k_{a_{m+1}}\leq i$, such that $a_{m+1}$ can be 
inserted into $G'$ by first inserting 
an additional item for type $v_j$ and then an additional item  for type $v_i$ 
at position $k_{a_{m+1}}$. The existence of this position follows by the assumption of the claim. 
%
The insertion of the two new items  at position $k_{a_{m+1}}$  
for arc  $a_{m+1}=(v_i,v_j)$ defines the following sets of arcs.
\begin{enumerate}
\item $A_{1,m+1}=\{(v_i,v_j)\}$ 
\item Set $A_{2,m+1}$ of all arcs from vertices
corresponding to items of $q'$ on position $\leq k_{a_{m+1}}$ to $v_i$ and to $v_j$ and 
all arcs from  $v_i$ and from $v_j$
to all vertices
corresponding to items of $q'$ on position $> k_{a_{m+1}}$.
\item Set $A_{3,m+1}$ of all arcs from vertices
corresponding to new items of $q_m$ on positions $k_{a_{j}}$, $j\leq m$,  to $v_i$ and to $v_j$ and 
all arcs from  $v_i$ and from $v_j$
to all vertices
corresponding to new items of $q_m$ on positions $k_{a_{j}}$, $j\geq m$.
\end{enumerate}

It remains to show that $A_{1,m+1}\cup A_{2,m+1} \cup A_{3,m+1} \subseteq A$. $A_{1,m+1} \subseteq A$ holds by definition and 
$A_{2,m+1} \subseteq A$ holds by  the assumption of the claim.
Set $A_{3,m+1}$ is equal to the union of all sets of arcs between 
the vertices corresponding
to the two new inserted items for $a_{m+1}$ and the vertices corresponding
to the two new inserted items for every $a_{j}$, $j\leq m$ following the order of the items.
As our case distinction above implies that all arcs between the new items inserted at any
$a_{j}$, $j\leq m$ and the new items inserted at  $a_{m+1}$  are in $A$ the union of all
these arcs $A_{3,m+1}$ is also in $A$.

So we can insert $a_{m+1}$ by first inserting 
an additional item for type $v_j$ and then an additional item for type $v_i$ 
at position $k'_{a_{m+1}}$, where  $k'_{a_{m+1}}$ is the position in $q_m$ which corresponds to 
position $k_{a_{m+1}}$ in $q'$.

Thus we have shown that $\g(\{q_m\})=G\in S_{1,\ell}$ for some $\ell\geq 2$.
\end{proof}

Assume that $G\not\in S_{1,\ell}$ for every $\ell\geq 2$. By Claim \ref{cl}
there is some arc $(v_i,v_j)\in A-A'$ such that for every position $k$, $j<k\leq i$  inserting 
an additional item for type $v_i$ and an additional item for type $v_j$ 
at position $k$ defines an arc which is not in $A$. 
That is, for every position $k$, $j < k \leq i$, in $q'$ 
there exists some $m'$, $j< m ' \leq k$, such that it holds $(v_{m'},v_j)\not\in A$
{\em or} there exists some $m''$, $k\leq m''< i$, such that it holds $(v_i,v_{m''})\not\in A$.
By the transitivity of $\co G$ it follows that there is one 
position $k$, $j < k \leq i$, in $q'$ 
such that there exists some $m'$, $j< m ' \leq k$, such that it holds $(v_{m'},v_j)\not\in A$
{\em and} there exists some $m''$, $k\leq m''< i$, such that it holds $(v_i,v_{m''})\not\in A$.
%

\medskip
If  $\co G=(V,A^c)$ is the complement digraph of $G$
we know that
\begin{equation}
(v_{m'},v_j)\in A^c \text{ and } (v_i,v_{m''})\in A^c. \label{b2miss}
\end{equation}

Since $m'\leq m''$ we know that $(v_{m'},v_{m''})\in A$. Further we know
that  $(v_{m''},v_{m'})\in A$, since otherwise $(v_{m''},v_{m'})\in A^c$, property
(\ref{b2miss}), and the transitivity of $\co G$ would imply that  $(v_i,v_j)\in 
A^c$ which is not possible. Thus
we know that
\begin{equation}
(v_{m'},v_{m''})\not \in A^c \text{ and } (v_{m''},v_{m'})\not\in A^c. \label{b12miss}
\end{equation}

Further the arcs $(v_j,v_{m'}),(v_j,v_{m''}),(v_{m'},v_i),(v_{m''},v_i)$ belong to
$A'\subseteq A$ and thus 
\begin{equation}
(v_j,v_{m'})\not \in A^c, \text{ } (v_j,v_{m''})\not\in A^c, \text{ } (v_{m'},v_i)\not\in A^c\text{ and } (v_{m''},v_i)\not\in A^c. \label{b123miss}
\end{equation}

If $(v_{i},v_{m'})\in A^c$ or $(v_{m''},v_{j})\in A^c$
then (\ref{b2miss}) and the transitivity of $\co G$ would imply 
that $(v_i,v_j)\in A^c$, thus we know
\begin{equation}
(v_{i},v_{m'})\not \in A^c \text{ and } (v_{m''},v_{j})\not\in A^c. \label{b1234miss}
\end{equation}

Properties
(\ref{b2miss})-(\ref{b1234miss}) imply that the subdigraph  $(\{v_i,v_j,v_{m'},v_{m''}\},\{(v_i,v_{m''}),(v_{m'},v_j)\})$
of $\co G$  leads to a $2\overrightarrow{P_2}$, which implies that $G$ contains a $\co(2\overrightarrow{P_2})$.
%
%
%
%
\end{proof}

\begin{table}[h!]
$$
\begin{array}{|r|c|c|c|c|c|c|c|c|}
   &      & v_j &       & v_k &       & v_i &        \\
\hline
1.  &  v_{j'} v_{k'} v_{i'}  &   &   &&&&       \\
2.  &  v_{j'} v_{k'}  & v_{i'}   &    &&&&    \\
3.  &  v_{j'} v_{k'}  & & v_{i'}   &  &&&     \\
4.  &  v_{j'} v_{k'}  & & & & v_{i'}   & &    \\
5.  &  v_{j'} v_{k'}  & & & && v_{i'}   &     \\
6.  &  v_{j'} v_{k'}  & & & && & v_{i'}      \\
7.  &  v_{j'} & & v_{k'} v_{i'} & & &&        \\
8.  &  v_{j'} & & v_{k'} & & v_{i'} & &       \\
9.  &  v_{j'} & & v_{k'} &&&v_{i'} &        \\
10. &  v_{j'} &  & v_{k'} &&& & v_{i'}       \\
11. &  v_{j'} &  & &&v_{k'} v_{i'}  &&      \\
12. &  v_{j'} &  & &&v_{k'} &v_{i'}  &       \\
13. &  v_{j'} &  & &&v_{k'} &&v_{i'}      \\
14. &  v_{j'} &  & && &&v_{k'}v_{i'}       \\
15. &         &v_{j'} & v_{k'}v_{i'} & && &     \\
16. &         &v_{j'} & v_{k'} & &v_{i'}& &     \\
17. &         &v_{j'} & v_{k'} & &&&v_{i'}      \\
18. &         &v_{j'} && & v_{k'} v_{i'} &  &  \\
19. &         &v_{j'} && & v_{k'} &&v_{i'}      \\
20. &         &v_{j'} && &&& v_{k'} v_{i'}      \\
21. &         && v_{j'} v_{k'} v_{i'}  &      &&&   \\
22. &  && v_{j'} v_{k'}  &  & v_{i'} &   &   \\
23. &  && v_{j'} v_{k'}  &  & &v_{i'} &       \\
24. &  && v_{j'} v_{k'}  &  & & &v_{i'}      \\
25. &  && v_{j'}   &  &v_{k'}  v_{i'} &   &   \\
26. &  && v_{j'}   &  &v_{k'}  &v_{i'} &     \\
27. &  && v_{j'}   &  &v_{k'}  && v_{i'}      \\
28. &  && v_{j'}   &    && &v_{k'}v_{i'}       \\
29. &  &&    &    &v_{j'}v_{k'}v_{i'} & &       \\
30. &  &&    &    &v_{j'}v_{k'} &v_{i'} &       \\
31. &  &&    &    &v_{j'}v_{k'} &&v_{i'}        \\
32. &  &&    &    &v_{j'} &&v_{k'}v_{i'}       \\
33. &  &&    &    & &v_{j'}&v_{k'}v_{i'}        \\
34. &  &&    &    & &&v_{j'}v_{k'}v_{i'}        \\
\hline
\end{array}
$$
\caption{Cases within the proof of Claim \ref{cl}\label{cases-tab}}
\end{table}

\begin{corollary}\label{theo-no1a} Every digraph  in $S_{k,\ell}$ can be obtained by
the union of at most $k$ many digraphs from $S_{1,2}$ and thus by the union of at most $k$ many 
$\{\co(2\overrightarrow{P_2}),2\overleftrightarrow{K_1},\overrightarrow{C_3},D_4\}$-free
digraphs. 
\end{corollary}

\begin{proposition}\label{le-find}
Let $G\in S_{1,2}$, then a set $Q$ of one sequence $q$, such
that $G=\g(Q)$ can be found in time $\bigo(|V|+|A|)$.
\end{proposition}

\begin{proof}
Let $G=(V,A)\in S_{1,2}$ and $q=[]$. We perform the following steps until $G=(\emptyset,\emptyset)$.
\begin{itemize}
\item Choose $v\in V$ such that $(v,u)\in A$ for all $u\in V-\{v\}$ and append $v$ to $q$.

\item Remove all arcs $(v,u)$ from $A$. 

\item If $\ideg(v)=\odeg(v)=0$, remove $v$ from $V$.

\item If there are vertices $u$ such that $\ideg(u)=\odeg(u)=0$, remove $u$ from $V$
and append $u$ to $q$.
\end{itemize}

Next we show that for every $G=(V,A)\in S_{1,2}$ this method leads to a sequence $q$, such that 
$G=\g(\{q\})$ and for every type in $\types(q)$ there are at most two items in $q$.

In order to perform the algorithm there has to be an ordering $v_1,\ldots,v_n$ of $V$ such that 
for $1\leq i <n$ vertex $v_i$ has maximum possible outdegree  in subdigraph 
obtained by removing the outgoing arcs of $v_1,\ldots,v_{i-1}$ and thereby 
created isolated vertices from $G$.
Since $G \in S_{1,2}$ there is a sequence $q'$ such that $G=\g(\{q'\})$. The order in which
the types corresponding to the vertices of $V$ appear in subsequence $F(q')$, defined in the proof of Theorem \ref{s12}, 
ensures the existence of such an ordering.


By performing the operations of the algorithm  along
such a vertex ordering we observe that in every
iteration only the outdegree of $v$ is changed and only vertices with defined incoming and outgoing arcs 
are removed. Thus by the definition of the sequence digraph 
it holds $G=\g(\{q\})$. Since every vertex which has only outgoing or only incoming arcs will be inserted 
once into $q$ and every vertex which has outgoing and incoming arcs
will be inserted at most twice into $q$ this sequence fulfils the properties stated in the theorem.
\end{proof}

\begin{example} 
We apply the method given in the proof of
Proposition \ref{le-find} on the digraph $D_6$ in Table \ref{F-co2}
and obtain $q=[b,a,c,b,d,c]$.
This even leads to a shorter sequence than we used in Example \ref{notta}(\ref{notta2}.).
\end{example}

\subsection{Sequence Digraphs and Directed Co-Graphs}\label{sec-co}

In Section \ref{sec-1-1} we introduced the operations $\text{IV},\text{OD},\text{ID}$.
The set of all digraphs which can be defined by these operations $\mathcal{G}_{\{\text{IV},\text{OD},\text{ID}\}}$
is defined as the set of {\em oriented threshold graph}, see \cite{Boe15}.
By Theorem \ref{s11} the class  $S_{1,1}$ is a subclass of oriented threshold graphs.
Digraphs in $S_{1,2}$  can be bidirectional complete (see Example \ref{le-cl}(\ref{le-cl-1})), 
which is not
possible for oriented threshold graphs by their definition. This motivates to consider their
relation to the more general class of directed co-graphs.

Let $G_1=(V_1,E_1)$ and $G_2=(V_2,E_2)$ be two vertex-disjoint directed graphs. 
\begin{itemize}
\item
The {\em disjoint union} of $G_1$ and $G_2$, 
denoted by $G_1 \oplus G_2$, 
is the digraph with vertex set $V_1 \cup V_2$ and 
arc set $E_1\cup E_2$. 

\item
The {\em series composition} of $G_1$ and $G_2$, 
denoted by $G_1\otimes G_2$, 
is the digraph with vertex set $V_1 \cup V_2$ and 
arc set $E_1\cup E_2\cup\{(u,v),(v,u)~|~u\in V_1, v\in V_2\}$. 

\item
The {\em order composition} of $G_1$ and $G_2$, 
denoted by $G_1\oslash  G_2$, 
is the digraph with vertex set $V_1 \cup V_2$ and 
arc set $E_1\cup E_2\cup\{(u,v)~|~u\in V_1, v\in V_2\}$. 
\end{itemize}

The class of directed co-graphs has been defined recursively by Bechet et al.~in \cite{BGR97}.
\begin{enumerate}[(i)]
\item Every digraph on a single vertex $(\{v\},\emptyset)$, 
denoted by $\bullet$, is a directed co-graph.

\item If $G_1$ and $G_2$ are directed co-graphs, 
then $G_1\oplus G_2$  is a directed co-graph.

\item  If $G_1$ and $G_2$ are directed co-graphs, then 
$G_1 \otimes G_2$  is a directed co-graph. 

\item  If $G_1$ and $G_2$ are directed co-graphs, then 
$G_1\oslash G_2$  is a directed co-graph.
\end{enumerate}

Directed co-graphs can be characterized by the  forbidden induced
subdigraphs shown in Table \ref{F-co}.

\begin{theorem}[\cite{CP06}]\label{co-for}
Digraph $G$ is a directed co-graph if and only if 
$G$ is $\{D_1, \ldots, D_8\}$-free.
\end{theorem}

\begin{table}[ht!]
\begin{center}
\begin{tabular}{cccccccc}

\epsfig{figure=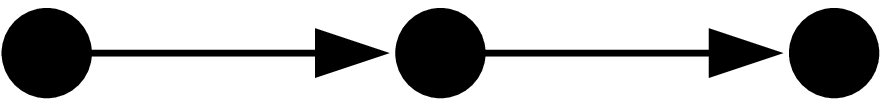,width=2.4cm} &&\epsfig{figure=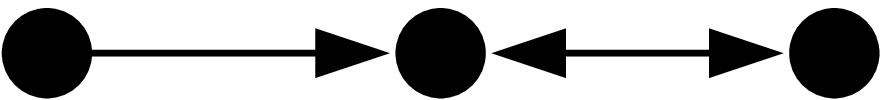,width=2.4cm}&&\epsfig{figure=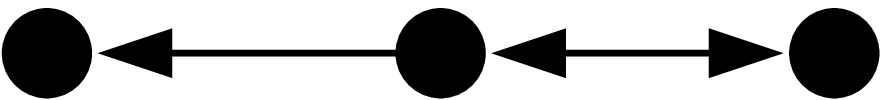,width=2.4cm}&&\epsfig{figure=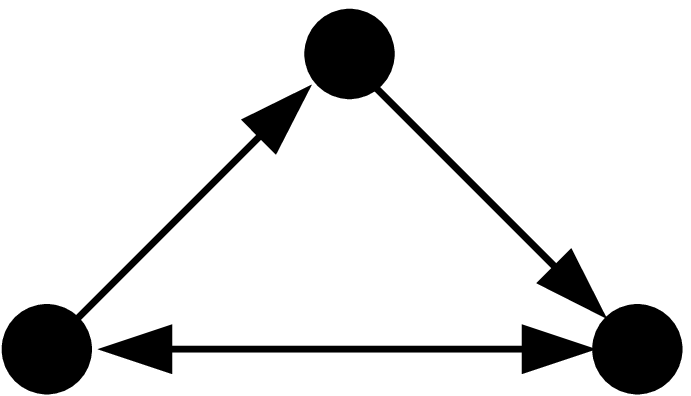,width=1.8cm} &\\
 $D_1$   &  &    $D_2$   &  &   $D_3$   &  &   $D_4$   &  \\ 
&&&&&&&\\
\epsfig{figure=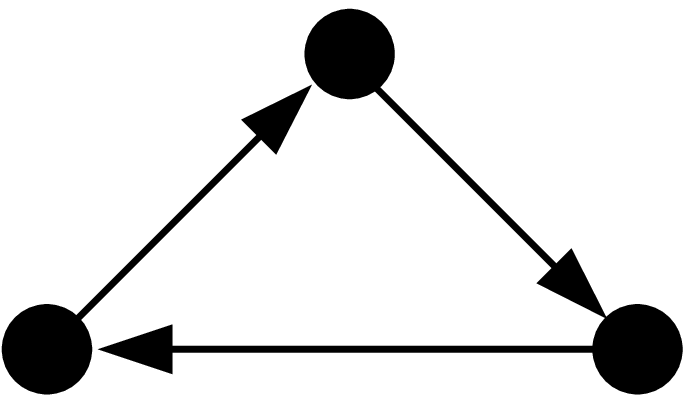,width=1.8cm} &&\epsfig{figure=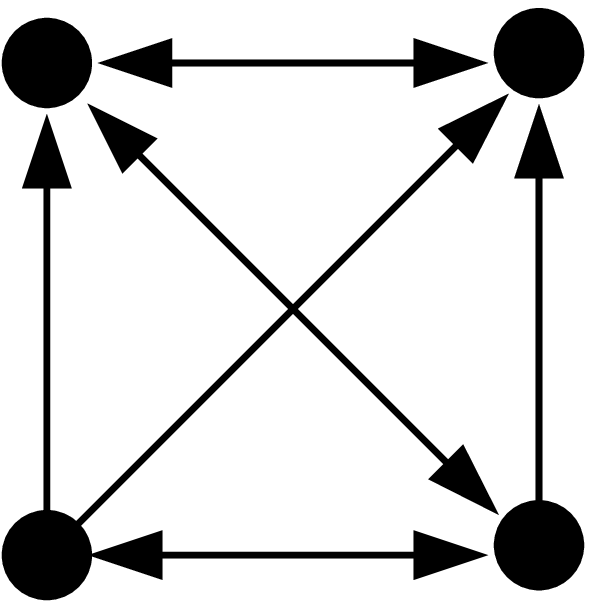,width=1.6cm}&&\epsfig{figure=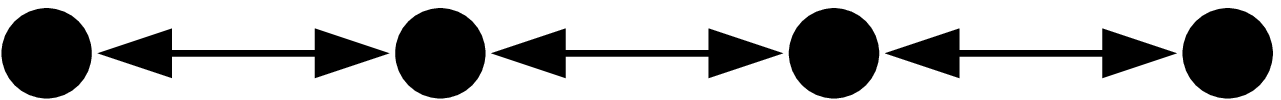,width=3.2cm}&&\epsfig{figure=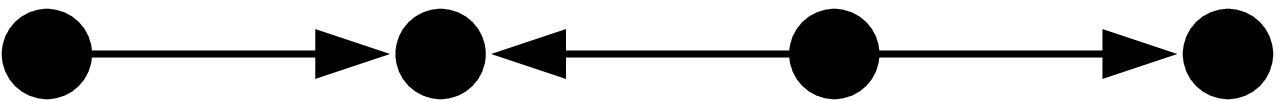,width=3.5cm}&\\
 $D_5$   &  &    $D_6$   &  &   $D_7$   &  &   $D_8$   &  \\ 

\end{tabular}
\end{center}
\caption{The eight forbidden induced subdigraphs for directed co-graphs.}
\label{F-co}
\end{table}



In order to state the relation between the classes $S_{1,\ell}$ and directed co-graphs
we classify the eight forbidden induced
subdigraphs shown in Table \ref{F-co} w.r.t.\ the classes  $S_{1,\ell}$.

\begin{lemma}\label{classi}
\begin{enumerate}
\item For every $\ell\geq 1$ it holds $D_1,D_2,D_3,D_4,D_5,D_7,D_8\not\in S_{1,\ell}$.

\item $D_6\in S_{1,2}-S_{1,1}$
\end{enumerate}
\end{lemma}

\sproof
\begin{enumerate}
\item For $i\in\{1,2,3,7,8\}$ graph $\un(D_i)$ is not complete and thus these digraphs $D_i$ 
do not belong to $S_{1,\ell}$ for some
$\ell\geq 1$
by Lemma \ref{le-ud}. Digraphs $D_4$ and $D_5$ are not in $S_{1,\ell}$ for some
$\ell\geq 1$ by
a simple case distinction. 

\item $D_6\in S_{1,2}$ by the sequence given in Example \ref{notta}(\ref{notta2})
and $D_6\not\in S_{1,1}$ by Theorem \ref{s11}(\ref{s11b}).\eproof
\end{enumerate}

Thus the classes directed co-graphs and $S_{1,2}$ are incomparable 
w.r.t.~inclusions, by the example  $D_6$  and $\co(2\overrightarrow{P_2})$.

By Lemma \ref{classi} we obtain the next lemma.

\begin{lemma}\label{le-d6-co}
For every $\ell\geq 1$ every $D_6$-free digraph in $S_{1,\ell}$ 
is a directed co-graph.
\end{lemma}

By the following Example  there are directed co-graphs which
are not in $S_{1,\ell}$. Thus 
the reverse direction of Lemma \ref{le-d6-co} does not hold true.

\begin{example}
For every $n\geq m\geq 1$ it holds  $\overleftrightarrow{K_{n,m}}\in S_{n\cdot m, 2n}$.
But for $k<nm$ or $\ell <2n$ by Lemma \ref{c3}
it holds $\overleftrightarrow{K_{n,m}}\not\in S_{k,\ell}$.
\end{example}


Theorem \ref{co-for}, Theorem \ref{s11}, and Theorem \ref{s12} imply the following results.

\begin{proposition}\label{le-d7-co}
Let $\ell\geq 1$ and $G\in S_{1,\ell}$. 
Then $G$ is a directed co-graph if and only if 
$G$ is $D_6$-free.
\end{proposition}

\begin{proposition}\label{le-d8-co}
Let $G$ be a directed co-graph. Then $G\in S_{1,1}$ if and only if 
$G$ is $\{2\overleftrightarrow{K_1},\overleftrightarrow{K_2}\}$-free
and $G\in S_{1,2}$ if and only if 
$G$ is $\{\co(2\overrightarrow{P_2}),2\overleftrightarrow{K_1}\}$-free.
\end{proposition}

Next we characterize the intersections of the sets $S_{1,\ell}$ and
directed co-graphs. For $\ell=1$ this leads to the set $S_{1,1}$
itself, since it is a subset of directed co-graphs.

\begin{proposition}\label{le-d9-co} 
Let $G$ be some digraph. 
Then
$G\in S_{1,1}$ and $G$ is a directed co-graph if and only if 
$G$ is $\{2\overleftrightarrow{K_1},\overleftrightarrow{K_2},\overrightarrow{C_3}\}$-free. 
Furthermore
$G\in S_{1,2}$ and $G$ is a directed co-graph if and only if 
$G$ is 
$\{\co(2\overrightarrow{P_2}),2\overleftrightarrow{K_1},\overrightarrow{C_3},D_4,D_6\}$-free.
\end{proposition}

While  $S_{1,1}$ is a subclass of the set of  oriented threshold graphs, 
the classes  $S_{1,2}$ and directed co-graphs (even without allowing the 
series composition) are not comparable
w.r.t. set inclusions. This follows by the fact, that $D_6\in S_{1,2}$ 
but not a directed co-graph. Further there are no constants $k$,$\ell$ 
such that every directed co-graph is in $S_{k,\ell}$. This can be verified 
by the directed co-graph $\overrightarrow{K_{n,m}}$ (see Example \ref{ex-knm}).
By Lemma \ref{c3} it holds  $\overrightarrow{K_{n,m}}\in S_{n\cdot m,\max(n,m)}$
but $\overrightarrow{K_{n,m}}\not \in S_{k,\ell}$ for $k<n\cdot m$ or $\ell<\max(n,m)$.

\subsection{Converse Sequence Digraphs and Complements of Sequence Digraphs}\label{sec-converse}

For some sequence $q_i=(b_{i,1},\ldots,b_{i,n_i})$ we define by
$q^{-1}_i=(b_{i,n_i},\ldots,b_{i,1})$ its {\em converse sequence} and
for some set $Q$ we define by $Q^{-1}=\{q^{-1}_i~|~q_i\in Q\}$  
the set of {\em converse sequences}.
Next we state a close relation between the sequence digraph
of some set $Q$ and of the converse digraph of the sequence digraph of 
$Q^{-1}$.

\begin{lemma}\label{inv}
Let $Q$  be a set of sequences. Then it holds $G=\g(Q)$ if and only
if $\con G=\g(Q^{-1})$.
\end{lemma}

This relation implies that
the classes $S_{k,\ell}$ are closed under taking converse digraphs.

\begin{proposition}\label{le-converse}
Let $G\in S_{k,\ell}$, then $\con G\in S_{k,\ell}$.
\end{proposition}


Next we want to address the relationship between the complement digraph of a 
sequence digraph $\g(Q)$ for some set $Q$
and the sequence digraph $\g(Q^{-1})$ of the converse sequences $Q^{-1}$. 
In general these digraphs are different. This can be 
verified by the two sequence digraphs considered in Example \ref{notta}.
Under certain conditions for some set of
sequences $Q$ the complement digraph of $\g(Q)$ can be obtained by the sequence digraph of 
the converse sequences $Q^{-1}$, which is useful since $Q^{-1}$ has the same number
of sequences and the same number of items per types as $Q$.


\begin{lemma}\label{rel-compl-c}
Digraph $G$ is a tournament if and only if $\co G=\con G$.
\end{lemma}

\begin{proof}
Digraph $G$ is a tournament, if and only if between every two vertices of $G$
there is exactly one arc, if and only if the converse digraph $\con G$ corresponds to the
complement digraph $\co G$. 
\end{proof}

\begin{lemma}\label{complb}
Let $G\in S_{k,\ell}$ defined by some set of sequences $Q$, such that 
graph $\g(Q)$ is a tournament, then $\co(\g(Q))= \g(Q^{-1})$.
\end{lemma}

\begin{proof}
Since $\g(Q)$ is a tournament by Lemma  \ref{rel-compl-c} it holds  $\co(\g(Q))=\con(\g(Q))$
and by  Lemma \ref{inv} it holds $\con(\g(Q))=\g(Q^{-1})$.
\end{proof}

Since $Q$ and $Q^{-1}$ are defined on the same number of sequences
and items of each type the following result holds.

\begin{corollary}\label{cor-comp}
Let $G\in S_{k,\ell}$ be a tournament, 
then $\co G\in S_{k,\ell}$.
\end{corollary}

Since the class $S_{1,1}$ is equivalent to the class
of  transitive tournaments (Theorem \ref{s11}),
the last two results hold for digraphs in $S_{1,1}$.

\begin{lemma}\label{compl}
Let $G\in S_{1,1}$ defined by some set $Q=\{q_1\}$ of one sequence, 
then $\co(\g(Q))= \g(Q^{-1})$.
\end{lemma}


\begin{proposition}\label{le-comp}
Let $G\in S_{1,1}$, then $\co G\in S_{1,1}$.
\end{proposition}

The results can not be generalized to graphs from $S_{1,2}$ 
and also not sets 
$Q$ on more than one sequence
since these are not tournaments in general.


Beside transitivity shown in Lemma \ref{le-transi} we observe that the complements 
digraphs of sequence digraphs are special directed acyclic graphs.

\begin{lemma}\label{le-ac}
Let $\ell\geq 1$ and  $G\in S_{1,\ell}$, then 
the complement digraph $\co G $ is acyclic.
\end{lemma}

\begin{proof}
Let $Q=\{q_1\}$ be some sequence which defines the sequence digraph $\g(Q)$.
If the complement digraph $\co(\g(Q))$ contains a cycle 
$(\{v_1,\ldots,v_n\},\{(v_1,v_2),\ldots, (v_{n-1},v_n),(v_n,v_1)\})$, $n\geq 2$,
then by Lemma \ref{le-transi} vertex set  $\{v_1,\ldots,v_n\}$
induces a bidirectional complete subdigraph in
digraph $\co(\g(Q))$. Thus the vertex set  $\{v_1,\ldots,v_n\}$
induces an edgeless subdigraph in $\g(Q)$ which contradicts Lemma \ref{le-ud}.
\end{proof}

It is well known that every DAG has at least one sink, i.e.~a vertex $v_1$ of
outdegree $0$, and one source, i.e.~a vertex $v_2$ of indegree $0$  by \cite{BG09}.
This means by Lemma \ref{le-ac} that for every $\ell\geq 1$  every digraph 
$G=(V,A)\in S_{1,\ell}$
has a vertex $v_1$ of outdegree $|V|-1$ and a vertex $v_2$ of 
indegree  $|V|-1$. For $\ell= 1$  this is known by  Theorem \ref{s11}.

\begin{lemma}\label{not-so-comp}
Let $\ell\geq 1$ and  $G=(V,A)\in S_{1,\ell}$, $u$ and $v$ be two
vertices from $V$, then 
the complement digraph $\co G$ has at most one of
the two arcs $(u,v)$ and $(v,u)$, i.e. $\co G$ is an oriented graph.
\end{lemma}

\begin{proof}
Let $G=(V,A)\in S_{1,\ell}$, $u$ and $v$ be two distinct
vertices from $V$. Then $G$ contains at least one of the two arcs $(u,v)$
and $(v,u)$ which implies that $\co G$ has at most one of
the two arcs $(u,v)$ and $(v,u)$.
\end{proof}

\begin{proposition}\label{le-xx}
Let $\ell\geq 1$ and $G\in S_{1,\ell}$, then 
the complement digraph $\co G$ has at most one non-trivial component 
which contains arcs.
\end{proposition}

\begin{proof}
If we assume that $\co G$ has more than one non-trivial component 
which contains arcs, then by Lemma \ref{not-so-comp} digraph 
$\co G$ contains 
$2\overrightarrow{P_2} =(\{a,b,c,d\},\{(c,a),(d,b)\})$, see
Table \ref{F-co2}, as an induced subdigraph. Lemma
\ref{le-p2free} leads to a contradiction.
\end{proof}

Digraphs in $S_{2,1}$ consist of the disjoint union of two 
digraphs in $S_{1,1}$ (see Lemma \ref{unionsk1}), thus the corresponding
complements have exactly one component. The 
complements of digraphs in $S_{2,2}$ can have more than one non-trivial component 
which contains arcs. This can
be verified by digraph $\co G\in S_{2,2}$  in the following example,
whose complement $G$ has two non-trivial components
which contains arcs.

\begin{example}\label{ex-2-join}
Let $G$ be the disjoint union of 
two digraphs $G'=(V',A')$ and $G''=(V'',A'')$ from  $S_{1,1}$. Then $G$ can be
defined by $Q=\{q_1,q_2\}$, where $q_1=[v'_1,v'_2,\ldots,v'_n]$ and $q_2=[v''_1,v''_2,\ldots,v''_m]$. 
W.l.o.g.~we assume the vertices in $V'$ and $V''$ to be enumerated in the order 
in which the graphs can be defined by inserting in-dominated vertices (cf.~Theorem \ref{s11}.(\ref{s11h})
or \ref{s11i}).
Thus it holds $G\in S_{2,1}$.

The complement digraph $\co G$ can be defined as follows. The digraph
complement of a digraph in $S_{1,1}$ also belongs to $S_{1,1}$ and the additional
arcs between $\co G'$ and $\co G''$ can be defined by 
two sequences, one for the arcs from $\co G'$ to $\co G''$  and one for the arcs from 
$\co G''$ to $\co G'$.
The arcs within $\co G'$ and $\co G''$  can be defined in both sequences. 
Formally $\co G$ can be
defined by $Q'=\{q'_1,q'_2\}$, where
$$q'_1=[v'_n,v'_{n-1},\ldots,v'_1, v''_m,v''_{m-1}\ldots,v''_1]$$
and
$$q'_2=[v''_m,v''_{m-1}\ldots,v''_1,v'_n,v'_{n-1},\ldots,v'_1].$$
Thus it holds $\co G\in S_{2,2}$.
\end{example}

Next we want to study for some digraph $G\in S_{k,\ell}$
whether it holds $\co G \in S_{k,\ell}$.
%
By Proposition \ref{le-comp}
the class $S_{1,1}$ is closed under complementations.
The next example shows that  the class $S_{1,2}$ is not  
closed under complementations and that the values of
$k$ and $\ell$ can grow arbitrary.

\begin{example}\label{ex-knm}
We consider 
$$\overrightarrow{K_{n,m}}=(\{v_1,\ldots,v_n,w_1,\ldots,w_m\},
\{(v_i,w_j)~|~ 1\leq i\leq n, 1\leq j \leq m\})$$
as an orientation of a complete
bipartite graph $K_{n,m}$.
The complement digraph of $\overrightarrow{K_{n,m}}$ can be defined by one sequence
$$q=[w_1,w_2,\ldots,w_m,w_1,w_2,\ldots,w_m,v_1,v_2,\ldots,v_n,v_1,v_2,\ldots,v_n]$$
which implies that the complement digraph of $\overrightarrow{K_{n,m}}$ is in $S_{1,2}$.
By Lemma \ref{c3} it holds  $\overrightarrow{K_{n,m}}\in S_{n\cdot m,\max(n,m)}$
but $\overrightarrow{K_{n,m}}\not \in S_{k,\ell}$ for $k<n\cdot m$ or $\ell<\max(n,m)$.
\end{example}

Since $S_{1,2}\subseteq S_{k,\ell}$ for $k\geq 1$ and $\ell\geq 2$ 
we know that the classes $S_{k,\ell}$ for every $k\geq 1$ and $\ell\geq 2$ 
are not closed under  complementations.
The next examples show that  $S_{2,1}$ is not closed under  complementations.

\begin{example}
\begin{enumerate}
\item Let $G_1$ be the disjoint union of two $1$-vertex digraphs in $S_{1,1}$.
Then $G_1\in S_{2,1}$. Further $\co G_1=\overleftrightarrow{K_2}\in S_{1,2}-S_{2,1}$.

\item  Let $G_n$ for $n\geq 2$ be the disjoint union of two $n$-vertex digraphs in $S_{1,1}$.
Then $G_n\in S_{2,1}$ and $\co G_n \in S_{2,2}-(S_{2,1}\cup S_{1,2})$, see 
Example \ref{ex-2-join}.
\end{enumerate}
\end{example}


\begin{lemma}\label{not-so-comp2ge}
Let $G\in S_{k,1}$ for $k\geq 2$, then $\co G\in S_{k(k-1),2(k-1)}$.
\end{lemma}

\begin{proof}
By Lemma \ref{unionsk1}
every digraph $G\in S_{k,1}$ is the disjoint union of $k$ digraphs $G_1,\ldots,G_k$ 
from $S_{1,1}$. 
For every $i$ it holds $\co G_i \in S_{1,1}$  (Proposition \ref{le-comp}).
Thus digraph $\co G$ can be obtained by $k$ digraphs from $S_{1,1}$
and all possible arcs between two of these $k$ graphs.
A sequence for  $\co G_i$ 
can be obtained by the reverse sequence of that for $G_i$ (Lemma \ref{compl}).
The arcs between two digraphs $\co G_i$  and $\co G_j$ can be 
created by two sequences.
For $k=2$ this is shown in  Example  \ref{ex-2-join}.
This leads to a set $Q$ of $k(k-1)$ sequences and for every type there are 
$2(k-1)$ items.
\end{proof}

\subsection{Simple Sequences}\label{sec-sim}

Let $Q=\{q_1,\ldots,q_k\}$ be a set of $k$ sequences. By Observation \ref{main_q} 
only the first and the last item of each type in every $q_i\in Q$ are important
for the arcs in the  corresponding digraph. 
Next we want to analyze how we even can use subsequences defined by the first
or last item of each type.
Let $F(q_i)$ be the subsequence of $q_i$
which is obtained from $q_i$ by removing all except the first item for each type
and $L(q_i)$ be the subsequence of $q_i$
which is obtained from $q_i$ by removing all except the last item for each type.\footnote{If
for some type there is only one item in $q_i$, then this item
remains in $F(q_i)$ and in $L(q_i)$.}

\begin{example}[$F(q_i)$, $L(q_i)$]\label{ex-simp}
For the sequence 
$q_1=[b,a,c,b,a]$ 
we obtain the two subsequences 
$F(q_1)=[b,a,c]$ 
and 
$L(q_1)=[c,b,a]$.
\end{example}

The given example shows that sequence digraph $\g(\{q_1\})$ differers from the
union of the two sequence digraphs  $\g(\{F(q_1)\})$ and $\g(\{L(q_1)\})$ 
only by arc $(a,b)$. In order to formalize this difference we define
special sequences.
Sequence $q_i\in Q$ 
is {\em simple}, if there do not exist $t,t'\in\types(q_i)$
such that  
\begin{equation}
\FI(q_i,t)<\FI(q_i,t')<\LA(q_i,t)<\LA(q_i,t'). \label{eq-si}
\end{equation}
Set $Q$ is {\em simple}, if every sequence 
$q_i\in Q$ is simple. 
We define $S'_{k,\ell}$
to be the set of all sequence digraphs defined by simple sets $Q$ 
on at most $k$ sequences that contain
at most $\ell$ items of each type in $\types(Q)$.
By the definition we
know for every  two integers $k\geq 1$ and $\ell\geq 1$
the following inclusions between these graph classes.
\begin{eqnarray}
S'_{k,\ell}    &\subseteq& S_{k,\ell}  \label{prop1a-s}\\
S'_{k,1}       &   =     & S_{k,1}     \label{prop2b-s} 
\end{eqnarray}

The sequence $q_1=[b,a,c,c,b,a]$ given in Example \ref{ex-simp} 
is not simple and sequence digraph $\g(\{q_1\})$ differers from the
union of the two sequence digraphs  $\g(\{F(q_1)\})$ and $\g(\{L(q_1)\})$.
The next theorem says that this is no coincidence.

\begin{theorem}\label{le-simple} 
A sequence $q_1$ is simple if and only if $\g(\{q_1\})$  can be obtained by
the union of  $\g(\{F(q_1)\})$ and $\g(\{L(q_1)\})$.
\end{theorem}

\begin{proof}
($\Rightarrow$)
Let $q_1$ be a simple sequence and $\g(\{q_1\})=(V,A)$ its sequence digraph. Further
let $\g(\{F(q_1)\})=(V,A')$ and $\g(\{L(q_1)\})=(V,A'')$ be the  sequence digraphs
of $F(q_1)$ and $L(q_1)$. We have to show that $A=A'\cup A''$. 

To show $A\subseteq A'\cup A''$ we consider subsequence $q'_1=M(q_1)$ 
(defined in Section \ref{SCgb}), which is also simple, defines the same
digraph as $q_1$, and  contains 
at most two items for every type. This implies that for every arc $(u,v)\in A$
vertices $u$ and $v$ correspond to the first or the last item of its type.
\begin{itemize}
\item If $(u,v)\in A$ since $\FI(q'_1,u)=b_i$  and $\FI(q'_1,v)=b_j$ for some $i<j$, then
$(u,v)\in A'$.
\item If $(u,v)\in A$ since $\LA(q'_1,u)=b_i$  and $\LA(q'_1,v)=b_j$ for some $i<j$, then
$(u,v)\in A''$.
\item If $(u,v)\in A$ since $\LA(q'_1,u)=b_i$  and $\FI(q'_1,v)=b_j$ for some $i<j$, then 
there is some $i'\leq i$ such that  $\FI(q_1,u)=b_{i'}$, which
implies that $(u,v)\in A'$. Further there is also    some $j'\geq j$ such that  $\LA(q_1,v)=b_{j'}$, which
implies that $(u,v)\in A''$.

\item If $(u,v)\in A$ since $\FI(q'_1,u)=b_i$  and $\LA(q'_1,v)=b_j$ for some $i<j$ 
we consider the following cases. If there is also some $i'$, $i\leq i' <j$, such that $\LA(q'_1,u)=b_{i'}$
then $(u,v)\in A''$. If there is also some $j'$, $i < j' \leq j$, such that $\FI(q'_1,v)=b_{j'}$
then $(u,v)\in A'$. 
If none of the two cases is fulfilled, then
$\FI(q'_1,v)<\FI(q'_1,u)<\LA(q'_1,v)< \LA(q'_1,u)$, which is not possible
within a simple sequence.
\end{itemize}
Further by the definition $\g(\{F(q_1)\})$ and $\g(\{L(q_1)\})$ are subgraphs
of $G$ and thus it holds 
that $A'\cup A''\subseteq A$.

($\Leftarrow$)
If $q_1$ is not  simple, there exist $t,t'\in\types(q_1)$
such that (\ref{eq-si}) holds for $i=1$. Thus arc $(t',t)$ is in $\g(\{q_1\})$ but
not in $\g(\{F(q_1)\})$ or $\g(\{L(q_1)\})$, which implies that
$\g(\{q_1\})$  can not be obtained by
the union of  $\g(\{F(q_1)\})$ and $\g(\{L(q_1)\})$.
\end{proof}

By the definition for every sequence $q_i$ the digraphs $\g(\{F(q_i)\})$ and $\g(\{L(q_i)\})$
both are in $S_{1,1}$ (see Theorem \ref{s11} for a precise characterization).

\begin{corollary}\label{theo-no} Let $G\in S'_{k,\ell}$, then $G$ can be obtained by
the union of $2k$ digraphs from $S_{1,1}$. 
\end{corollary}

In order to generalize this result to arbitrary sequences it remains to 
add the missing arcs  $(t',t)$ for every $t,t'\in\types(q_1)$
such that (\ref{eq-si}) holds by using further sequences. This is not
possible within one additionally sequence for each non-simple sequence.

%


%


\section{Directed Path-width of Sequence Digraphs}\label{sec-problems}

\subsection{Directed Path-width}\label{sec-dp}

According to Bar{\'a}t \cite{Bar06}, the notion of directed path-width was
introduced by Reed, Seymour, and Thomas around 1995 and relates to directed
tree-width introduced by Johnson, Robertson, Seymour, and Thomas in
\cite{JRST01}. 
%
%
A {\em directed path-decomposition} of a digraph $G=(V,A)$
is a sequence $(X_1, \ldots, X_r)$ of subsets of $V$, called {\em bags},  such 
that the following three conditions hold true.
\begin{enumerate}[(dpw-1)]
\item $X_1 \cup \ldots \cup X_r ~=~ V$.
\item For each $(u,v) \in A$ there is a pair $i \leq j$ such that
  $u \in X_i$ and $v \in X_j$.
\item  If $u \in X_i$ and $u \in X_j$ for some  $u\in V$ and two 
indices $i,j$ with $i \leq j$, then $u \in X_\ell$ for all indices $\ell$ 
with $i \leq \ell \leq j$.
\end{enumerate}
The {\em width} of a directed path-decomposition ${\cal X}=(X_1, \ldots, X_r)$ 
is $$\max_{1 \leq i \leq r} |X_i|-1.$$ The {\em directed path-width} of $G$,
$\dpw(G)$ for short, is 
the smallest integer $w$ such that there is a directed path-de\-com\-po\-sition for 
$G$ of width $w$. 

A  directed path-decomposition ${\mathcal X} =(X_1, \ldots, X_r)$  is called {\em nice}, if the symmetric difference
of $X_{i-1}$ and $X_i$ contains exactly one element. A bag
$X_i = X_{i-1} \cup \{t\}$, $t \not\in X_{i-1}$, introduces $t$ and
is called {\em introduce bag}. A bag $X_i = X_{i-1} - \{t\}$,
$t \in X_{i-1}$, forgets $t$ and is called {\em forget bag}. 
By the definition of a path-decomposition, within a
nice directed path-decomposition  every graph vertex
is introduced and forgotten exactly once. Thus every nice 
path-decomposition has $r=2|V|+1$ bags.
It is not hard to see that any  directed path-decomposition can be refined into a
nice  directed path-decomposition without increasing the 
width \cite{Bod97}.

There is a close relation between the directed path-width of 
some digraph and the  (undirected)  path-width (see \cite{RS83}) 
of its underlying undirected graph.

\begin{lemma}[Lemma 1 of \cite{Bar06}]\label{le-c-bi}
Let $G$ be some  complete bioriented digraph, then $\dpw(G)= \pw(\un(G))$.
\end{lemma}



Determining whether the (undirected) path-width of some given (undirected) graph  is 
at most some given value $w$ is NP-complete \cite{KF79} 
even for bipartite graphs, complements
of bipartite graphs \cite{ACP87}, 
chordal graphs \cite{Gus93},
bipartite distance hereditary graphs \cite{KBMK93},  
and planar graphs with maximum vertex
degree 3 \cite{MS88}. 
Lemma \ref{le-c-bi} implies
that determining whether the directed path-width of some given digraph  is 
at most some given value $w$ is NP-complete even for digraphs whose underlying 
graphs lie in the mentioned classes.  For example 
determining whether the directed path-width of some given
digraph with maximum semi-degree $\Delta^0(G)=\max\{\Delta^-(D),\Delta^+(D)\}\leq 3$
is  at most some given value $w$ is NP-complete, which will be useful in 
Proposition \ref{hard-dpw-deg}.
While undirected path-width can be solved by an FPT-algorithm \cite{Bod96}, 
the existence of such an algorithm for directed path-width is still open.
%
The directed path-width of a digraph $G=(V,A)$ can be computed in time 
$\bigo(\nicefrac{|A|\cdot |V|^{2\dpw(G)}}{(\dpw(G)-1)!})$ 
by \cite{KKKTT16} and in time 
$\bigo(\dpw(G)\cdot|A|\cdot |V|^{2\dpw(G)})$ 
by \cite{Nag12}. This leads to  XP-algorithms
for directed path-width w.r.t.~the standard parameter
and implies that for each constant $w$, it is decidable in polynomial time whether a given 
digraph has directed path-width at most $w$. 
Further in \cite{KKT15} it is shown how to decide whether the directed
path-width of an $\ell$-semicomplete digraph is at most $w$ 
in time $(\ell+2w+1)^{2w}\cdot |V|^{\bigo(1)}$.

The next lemma follows by the definition of converse digraphs and
path-decompositions.

\begin{lemma}\label{le-pw} Let $G$ be a digraph.
Sequence $(X_1, \ldots, X_r)$ is a directed path-decomposition 
for $G$ if and only if sequence $(X_r, \ldots, X_1)$ 
is a directed path-decomposition of $\con G$. 
\end{lemma}

\begin{lemma}\label{le-pw2} 
Let $G$ be some digraph, then $\dpw(G)=\dpw(\con G)$. 
\end{lemma}

The directed path-width 
can change when taking the complement digraph, this is not possible
if we restrict to tournaments.

\begin{lemma}\label{pw-g-gc}
For every tournament $G$ it holds $\dpw(G)=\dpw(\co G)$.
\end{lemma}

\begin{proof}
By Lemma \ref{le-pw2} we know that $\dpw(G)=\dpw(\con G)$ and for
tournaments $G$  by \ref{rel-compl-c} it 
holds $\dpw(\con G) = \dpw(\co G)$.
\end{proof}

\subsection{Hardness of Directed Path-width on Sequence Digraphs}\label{sec-h}

Next we give some conditions on the sequences in $Q$
such that for the corresponding digraph $\g(Q)$ 
computing its directed path-width is NP-hard.

\begin{proposition}\label{th-hard-sq-g}
Given some set $Q$ on $k$ sequences  such that
$n_i=2$ for $1\leq i \leq k$ and some integer $p$, then the problem
of deciding whether $\dpw(\g(Q))\leq p$ is NP-complete.
\end{proposition}

\begin{proof}
The stated problem is  in NP since we can restrict to nice directed
path-decompositions having a polynomial number of bags. 
To show the NP-hardness by a
reduction from the directed path-width problem we transform instance 
$(G,p)$ in linear time into instance $(\q(G),p)$ for the
stated problem. The correctness follows by Observation  \ref{prop}.
\end{proof}

\begin{proposition}\label{hard-dpw-deg}
Given some set $Q$ such that  
$d_Q=3$ or $c_Q=5$ and some integer $p$, then the problem
of deciding whether  $\dpw(\g(Q))\leq p$ is NP-complete.
\end{proposition}

\begin{proof}
To show the NP-hardness by a
reduction from the directed path-width problem for digraphs $G$ such that
$\max(\Delta^-(G),\Delta^+(G))\leq 3$, we transform instance 
$(G,p)$ in linear time into instance $(\q(G),p)$ for the
stated problem.  The correctness follows by Lemma  \ref{lemma-c-d}.
\end{proof}

Similar results can be shown for related parameters such as
directed cut-width \cite{CFS12}.

\subsection{Polynomial Cases of Directed Path-width on Sequence Digraphs}\label{sec-p}

Next we consider the directed path-width of sequence digraphs for  $k=1$ or $\ell=1$.

\begin{proposition}\label{th-pw1}
Let $G\in S_{k,1}$, then $\dpw(G)=0$.
\end{proposition}

\begin{proof}
By Theorem \ref{s1k}
every digraph in $S_{k,1}$ is the disjoint union of $k$
digraphs in $S_{1,1}$. By Theorem \ref{s11} every digraph in $S_{1,1}$ 
is acyclic and thus has directed path-width $0$.
\end{proof}

On the other hand, there are no constants $k$,$\ell$ 
such that every digraph of directed path-width $0$ is in $S_{k,\ell}$. 
This can be verified 
by an orientation $T'$ of a tree $T$ on $k$ edges and $\Delta(T)=\ell$.
Then $\dpw(T')=0$ and by Lemma \ref{c3} it holds  $T'\in S_{k,\ell}$
but $T'\not \in S_{k',\ell'}$ for $k'<k$ or $\ell'<\ell$.

For digraphs in $S_{1,2}$ the directed path-width can be arbitrary 
large, since this class includes all bidirectional complete 
digraphs. We can compute this value as follows.
Let $Q=\{q\}$. For type $t \in \types(q)$ let $I_t = [\FI(q,t), \LA(q,t)]$ be the
interval representing $t$, and let $I_q = \{I_t \mid t \in \types(q)\}$
be the set of all intervals for sequence $q$. Let $I(q) = (V,E)$
be the interval graph where $V = \types(q)$ and $E = \{\{u,v\} \mid
u \neq v,~ I_u \cap I_v \neq \emptyset, I_u,I_v\in I_q\}$, see Figure \ref{F03}.

%

\begin{figure}[hbtp]
\centerline{\epsfxsize=.75\textwidth \epsfbox{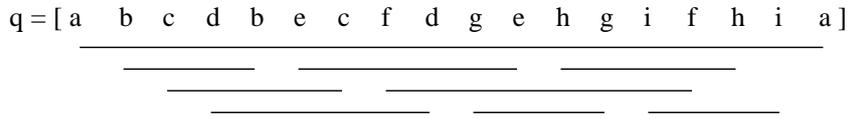}}
\caption{Sequence $q=[a,b,c,d,b,e,c,f,d,g,e,h,g,i,f,h,i,a]$ represented
  as an interval graph $I(q)$.}
\label{F03}
\end{figure}

\begin{proposition}\label{th-pw2}
Let $G\in S_{1,2}$ defined by some set $Q=\{q_1\}$ of one sequence, 
then  $\dpw(G)=\omega(I(q))-1=\pw(I(q))$.
\end{proposition}

\begin{proof}
It holds $\dpw(G)\leq\omega(I(q))-1$ by an obvious directed path-decomposition
along $I(q)$. Further for every integer $r$ the set $I(r)=\{I_t ~|~ r\in I_t\}$ defines a 
complete subgraph $K_{|I(r)|}$ in $I(q)$ 
and also a bidirectional complete subdigraph $\overleftrightarrow{K_{|I(r)|}}$ in $G$.
Thus it holds 
 $\dpw(G)\geq\omega(I(q))-1$.
The second equality holds since
the (undirected) path-width of an interval graph is equal to the size of 
a maximum clique \cite{Bod98}.
\end{proof}


In contrast to Proposition \ref{hard-dpw-deg} sets $Q$ where $d_Q=1$
can be handled in polynomial time.

\begin{proposition}\label{th-p-dq1}
Given some set $Q$ such that  
$d_Q=1$ and some integer $p$, then the problem
of deciding whether $\dpw(\g(Q))\leq p$
can be solved in time $\bigo(|\types(Q)|^2+n)$.
\end{proposition}

\begin{proof}
Let $Q=\{q_1,\ldots,q_k\}$. 
If $d_Q=1$ the vertex sets $V_i=\types(q_i)$ are disjoint. That is,
$\g(Q)$ is the disjoint union of digraphs in $S_{1,2}$
for which the directed path-width can be computed in 
time $\bigo(\sum_{i=1}^k|\types(\{q_i\})|^2+n_i)\subseteq \bigo(|\types(Q)|^2+n)$ by Proposition \ref{th-pw2}.
\end{proof}


\subsection{An XP-Algorithm for Directed Path-width}\label{sec-dpw-a}

We next give an XP-algorithm for directed path-width w.r.t.\ the parameter $k$, 
which implies that for every constant $k$ for a given set $Q$ on at most $k$ sequences the value
$\dpw(\g(Q))$ can be computed in polynomial time.
The main idea is to discover an optimal directed path-decomposition 
by scanning the $k$ sequences left-to-right and keeping in a state 
the numbers of scanned items of every sequence 
and a certain  number of active types.

Let $Q=\{q_1,\ldots,q_k\}$ be a set of $k$ sequences. Every $k$-tuple $(i_1,\ldots,i_k)$
where $0\leq i_j\leq n_j$ for $1\leq j \leq k$ is a {\em state} of $Q$.
State $(0,0, \ldots, 0)$ is the {\em initial state} and
$(n_1,\ldots,n_k)$ is the {\em final state}.
The {\em state digraph} $\s(Q)$ for a set $Q$ has a vertex for each possible state. 
There is an arc from vertex $u$
labeled by $(u_1, \ldots, u_k)$ to vertex $v$ labeled by $(v_1, \ldots, v_k)$ if
and only if $u_i = v_i - 1$ for exactly one element of the vector
and for all other elements of the vector $u_j = v_j$.
Let  $(i_1,\ldots,i_k)$ be a state of $Q$. We define 
$L(i_1,\ldots,i_k)$ to be the set of all items on the positions  $1,\ldots, i_j$ for $1\leq j\leq k$
and $R(i_1,\ldots,i_k)$ is 
the set of all items on the remaining positions $i_j+1,\ldots,n_j$
for $1\leq j\leq k$. 
Further let $M(i_1,\ldots,i_k)$ be the set of all items 
on the positions  $i_j$ for  $1\leq j\leq k$ such that there is exactly
one type of these items in $Q$.
Obviously, for every state $(i_1,\ldots,i_k)$ it holds that
\begin{itemize}
\item
$L(i_1,\ldots,i_k)\cup R(i_1,\ldots,i_k)$ leads to a disjoint partition of the
items in $Q$ and
\item
$M(i_1,\ldots,i_k)\subseteq L(i_1,\ldots,i_k)$.
\end{itemize}

Further each vertex $v$ of the state digraph is labeled by
the value $f(v)$.
This value is the number of types $t$ such that either
there is at least one item of type $t$ in $L(v)$ and at least one
item of type $t$ in $R(v)$ or there is one item of type $t$ 
in $M(v)$. Formally we define 
$\ac(v)=\{t\in\types(Q)~|~ b\in L(v), t(b)=t, b'\in R(v), t(b')=t\} \cup\{t\in\types(Q)~|~ b\in M(v), t(b)=t\}$
and $f(v)=|\ac(v)|$.
Obviously for the initial state $v$ it holds $|\ac(v)|=0$.
Since the state digraph $\s(Q)$ is a 
directed acyclic graph
we can compute all values $|\ac(v)|$ using
a topological ordering $topol$ of the vertices. 
Every arc $(u,v)$ in  $\s(Q)$ represents one item $b_{i,j}$ if item $b_{i,j-1}\not\in M(v)$ 
and two items $b_{i,j}$ and $b_{i,j-1}$ if item $b_{i,j-1}\in M(v)$  of some types $t(b_{i,j})=t$ 
and $t(b_{i,j-1})=t'$ from some 
sequence $q_j$, thus
\begin{eqnarray}
  |\ACT((i_1, \ldots, i_{j-1}, i_j+1, i_{j+1}, \ldots, i_k))|  = 
   |\ACT((i_1, \ldots, i_{j-1}, i_j, i_{j+1}, \ldots, i_k))| + c_{j}
  \label{transStep}
\end{eqnarray}
where 
\[
   c_{j} = \left\{
     \begin{array}{rll}
       1, & \multicolumn{2}{l}{\mbox{if } \FI(q_j,t) = i_j+1 \mbox{ and } \FI(q_\ell,t) > i_\ell 
              ~~ \forall ~ \ell \neq j} \mbox{ and }\\
          & \mbox{not}(\FI(q_j,t')=\LA(q_j,t')=i_j  \mbox{ and } \LA(q_\ell,t')=0  ~~ \forall ~ \ell \neq j) \\ 
       0, & \multicolumn{2}{l}{\mbox{if } \FI(q_j,t) = i_j+1 \mbox{ and } \FI(q_\ell,t) > i_\ell
              ~~ \forall ~ \ell \neq j} \mbox{ and }\\
         & \FI(q_j,t')=\LA(q_j,t')=i_j  \mbox{ and } \LA(q_\ell,t')=0  ~~ \forall ~ \ell \neq j \\
      -1, & \multicolumn{2}{l}{\mbox{if } \LA(q_j,t) = i_j+1 \mbox{ and } \LA(q_\ell,t) \leq i_\ell
              ~~ \forall ~ \ell \neq j} \mbox{ and }\\
         & \mbox{not}(\FI(q_j,t')=\LA(q_j,t')=i_j  \mbox{ and } \LA(q_\ell,t')=0  ~~ \forall ~ \ell \neq j) \\ 
   -2, &  \multicolumn{2}{l}{\mbox{if } \LA(q_j,t) = i_j+1 \mbox{ and } \LA(q_\ell,t) \leq i_\ell
              ~~ \forall ~ \ell \neq j} \mbox{ and }\\
    & \FI(q_j,t')=\LA(q_j,t')=i_j  \mbox{ and } \LA(q_\ell,t')=0  ~~ \forall ~ \ell \neq j \\
       0, & \multicolumn{2}{l}{\mbox{otherwise.}}
     \end{array}
   \right.
\]
Please remember our technical definition of $\FI(q,t)$ and $\LA(q,t)$
from Section \ref{SCgb}
for the case that $t\not\in \types(q)$.
Thus,
the calculation of value $|\ac(i_1, \ldots, i_k)|$ for the vertex
labeled $(i_1, \ldots, i_k)$ depends only on already calculated values,
which is necessary in order to use dynamic programming.\footnote{For sets $Q$ such that 
the number of items for which there is no further item of the same type in $Q$ is small,
we suggest to modify $Q$ by inserting a dummy item of the same type at the position
after such items. This does not change the sequence digraph but increases the size of the
sequence digraph but allows to make a case distinct within three instead of five cases when 
defining $c_j$.}

\begin{example}[State digraph]\label{EX5}
Let $Q =\{q_1, q_2\}$ of the sequences $q_1 = [a,b,c,b,d,a]$
and $q_2 = [e,f,e,c]$. The sequence digraph $\g(Q)$ is given in
Figure \ref{FFS} and the state digraph $\s(Q)$ is 
given in Figure \ref{FFP}.
For every
vertex $v$ of  $\s(Q)$ the vector $h(v)$ in the vertex represents the state of $v$ and  
every arc is labeled with the type of the
item, that is considered in order to obtain the next state, and  
the number $f(v)$ is given on the left of the vertex that represents state $h(v)$.
\end{example}

\begin{figure}[hbtp]
\centering
\parbox[b]{.45\textwidth}{
\centerline{\epsfxsize=50mm \epsfbox{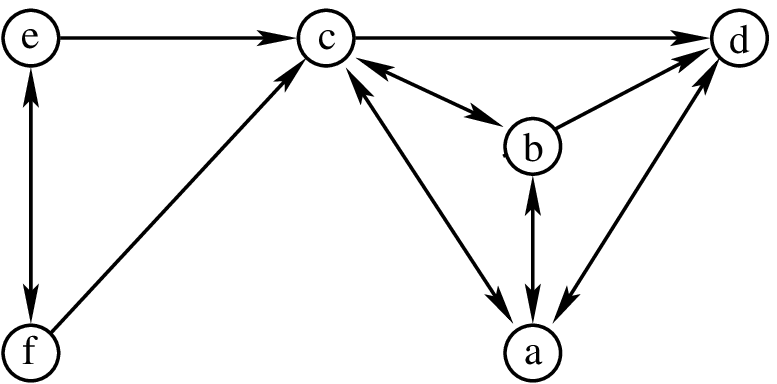}}
\caption{Sequence digraph  of Example \ref{EX5}.}
\label{FFS}
}
\hfill
\parbox[b]{.53\textwidth}{
\centerline{\epsfxsize=70mm \epsfbox{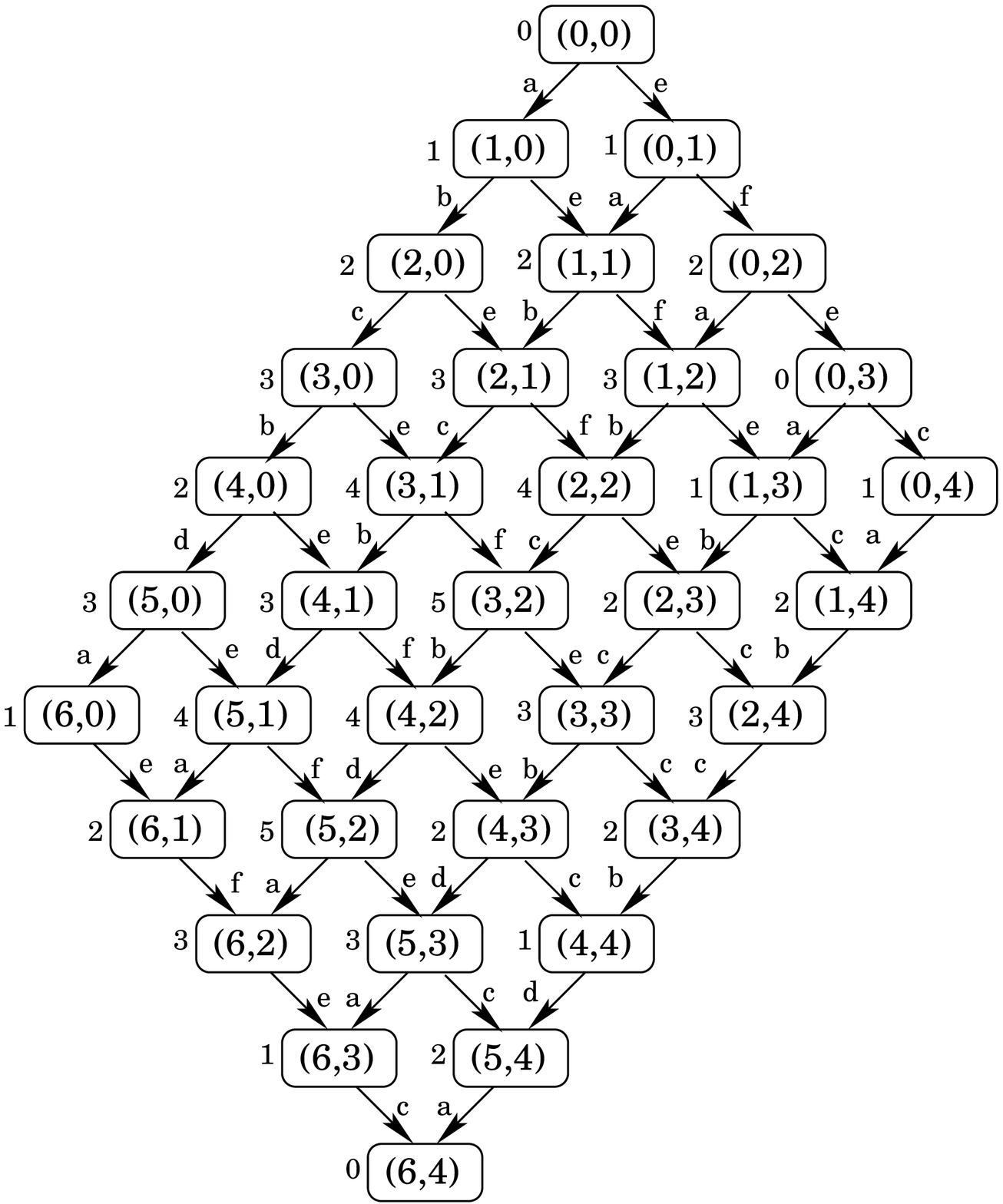}}
\caption{State digraph  of Example \ref{EX5}.}
\label{FFP}
}
\end{figure}

Let $\mathcal{P}(Q)$ the set of all paths from the initial state
to the final state in $\s(Q)$. Every path $P\in\mathcal{P}(Q)$ 
has $r=1+\sum_{i=1}^{k}n_i$ vertices, i.e. $P=(v_0,\ldots,v_r)$. 
First we observe that every path in $\mathcal{P}(Q)$
leads to a directed path-decomposition for  $\g(Q)$. 

\begin{lemma}\label{lem-paths}
Let $Q$ be some set of $k$ sequences and $(v_0,\ldots,v_r)\in\mathcal{P}(Q)$.
Then the sequence
$(\ac(v_1),\ldots,\ac(v_{r-1}))$ is a directed path-decomposition  for $\g(Q)$. 
\end{lemma}

\begin{proof}
Let $\g(Q)=(V,A)$ be the sequence digraph, $\s(Q)=(S,B)$ be the state digraph,
and $P=(v_0,\ldots,v_r)$ be an arbitrary path from the initial state
to the final state in $\s(Q)$. We next will verify the three conditions of
a directed path-decomposition. Since  $\types(Q)=V$ we identify
the vertices with their types.
\begin{itemize}
\item For every item $b$ of $Q$ there is a vertex $v_i\in S$, $1\leq i \leq r-1$, in $P$ such
that $b\in L(v_i)-L(v_{i-1})$. 
The definition of $\ac(v_i)$ implies that $t(b)\in \ac(v_i)$. It follows that  
$V=\ac(v_1)\cup\ldots \cup \ac(v_{r-1})$.

\item Let $(t,t')\in A$. By the definition of $\g(Q)$ there is one sequence $q_\ell$
such that there are two items $b_{\ell,m}$ of type $t$ and $b_{\ell,m'}$
of type $t'$ such that $m<m'$. Thus there is a vertex $v_i$ in $P$ such
that $b_{\ell,m}\in L(v_i)-L(v_{i-1})$ and a vertex $v_j$ in $P$ such
that $b_{\ell,m'}\in L(v_j)-L(v_{j-1})$. Thus 
$t \in \ac(v_i)$ and $t' \in \ac(v_j)$ for  $i \leq j$.

\item Let  $t \in \ac(v_i)$ and $t \in \ac(v_j)$ for
$i < j$. 
Thus there are four items $b$, $b'$, $b''$, and $b'''$ of type $t$ in $Q$
such that there is a vertex $v_i$ in $P$ such
that $b\in L(v_i)$ and $b'\in R(v_i)$ and a vertex $v_j$ in $P$ such
that $b''\in L(v_j)$  and $b'''\in R(v_j)$. Thus it
holds $b\in L(v_\ell)$ and $b'''\in R(v_\ell)$ for all $i\leq \ell\leq j$.
Thus $t\in \ac(v_\ell)$ for all indices $\ell$ 
with $i \leq \ell \leq j$.

\end{itemize}
Thus $(\ac(v_1),\ldots,\ac(v_{r-1}))$ is a directed path-decomposition  for $\g(Q)$.
\end{proof}

\begin{example}[Paths in $\s(Q)$]\label{EX5a} We consider two paths in  $\s(Q)$
of Example \ref{EX5}.
\begin{enumerate}
\item $P_1=((0,0),(0,1),(0,2),(1,2),(2,2),(3,2),(3,3),(3,4),(4,4),(5,4),(6,4))$
leads to $$(\{e\},\{e,f\},\{a,e,f\},\{a,b,e,f\},\{a,b,c,e,f\},\{a,b,c\},\{a,b\},\{a\},\{a,d\})$$
a directed path-decomposition of width $4$ for  $\g(Q)$.

\item  $P_2=((0,0),(1,0),(1,1),(1,2),(1,3),(2,3),(3,3),(4,3),(5,3),(6,3),(6,4))$
leads to $$(\{a\},\{a,e\},\{a,e,f\},\{a\},\{a,b\},\{a,b,c\},\{a,c\},\{a,c,d\},\{c\})$$
a directed path-decomposition of width $2$ for  $\g(Q)$.

\end{enumerate}
 \end{example}

Example \ref{EX5a} shows that $P\in\mathcal{P}(Q)$  can lead to directed
path-decompositions of non-optimal
width.
Lemma \ref{lem-paths} leads  to an upper bound on the directed
path-width of $\g(Q)$ using the state graph. The reverse direction is considered 
next.

\begin{lemma}\label{lem-paths2}
Let $Q$ be some set of $k$ sequences.
If there is a  directed path-decomposition of width $p-1$ for $\g(Q)$,
then there is a path $(v_0,\ldots,v_r)\in\mathcal{P}(Q)$ such that 
for every $1\leq i\leq r$ it holds $|\ac(v_i)|\leq p$.
\end{lemma}

\begin{proof}
We define a partial order on the types in $Q$.
Let $\alpha_Q: \types(Q)\rightarrow \NN$ such that for all  $t,t'\in \types(Q)$ 
it holds $\alpha_Q(t)<\alpha_Q(t')$,
if for all $q_i\in Q$ such that $t,t'\in \types(q_i)$
it holds $\FI(q_i,t)<\FI(q_i,t')$.

Let ${\mathcal X} =(X_1, \ldots, X_s)$  be a directed path-decomposition of width $p-1$  for $\g(Q)$.
We transform ${\mathcal X}$ into a nice  directed path-decomposition ${\mathcal X}' =(X'_1, \ldots, X'_{s'})$ 
such that the introduce bags are ordered w.r.t.~$\alpha_Q$. 
The order of the introduce bags it denoted by $(t_1,\ldots,t_m)$ which is 
used in algorithm {\sc Transform} given in Figure \ref{fig:algorithm3x}.
By the ordering of the introduce nodes for every $1\leq j \leq m$ in our for loop we find a 
sequence $q$ such that at the first position of the unprocessed part of $q$ 
there is a an item $b$ such that $t(b)=t_j$.
This allows us to remove this item $b$ and define a feasible successor state in $\s(Q)$.
The removal of further items $b$ from the first positions such that $t(b)\in S$ also
defines feasible successor states in $\s(Q)$. Thus the order in which the items
are removed from the sequences in line $(*)$ of algorithm
{\sc Transform} defines a path $(v_0,\ldots,v_r)\in\mathcal{P}(Q)$.

Further the order $\alpha_Q$ implies that for every state $v_i$, $1\leq i \leq r$
there is some bag $X'_j$, $1\leq j\leq s'$ such that $\ac(v_i)\subseteq X'_j$.
This implies that for every $1\leq i\leq r$ it holds $|\ac(v_i)|\leq p$.
\end{proof}

\begin{figure}[ht]
\hrule
{\strut\footnotesize \bf Algorithm {\sc Transform}($(t_1,\ldots,t_m)$)} 
\hrule
\footnotesize 
\smallskip
\begin{tabbing}
xxxx \= xxxx \= xxxx \= xxxx \= \kill
$S := ()$ \>\> {\it Comment:} $S$ contains later on the order of the types in which they are visited first\\
{\bf for} $j=1$ {\bf to} 1 $m$ \{ \\
\> {\bf if} ($\exists q \in Q: t(q.first()) = t_j$) \{ \\
\>\> $S.append(t_j)$ \\
\>\>{\bf while} ($\exists q \in Q: t(q.first()) \in S$) \\
\>\>\> remove the first item of sequence $q$ ~~~~$(*)$\\
\>\} \\

\}
\end{tabbing}
\normalsize
\hrule
\caption{Transforming nice directed path-decomposition of
$\g(Q)$ into a  path in $\mathcal{P}(Q)$.}
\label{fig:algorithm3x}
\end{figure}

\begin{example} \label{EX7}
We consider the set $Q = \{q_1, q_2, q_3\}$ of the sequences $q_1 = [a,a,d,e,d]$,
$q_2 = [c,b,b,d]$, and $q_3 = [c,c,d,e,d]$ from Example \ref{EX6}. 
Then ${\mathcal X} =(X_1,X_2 , X_3)$, where $X_1=\{a,c\}$, $X_2=\{b\}$, and $X_3=\{d,e\}$  is a 
directed path-decomposition of width $1$  for $\g(Q)$ shown in Figure \ref{F04}.

Our partial order $\alpha_Q$ has to fulfill $a<d<e$, $c<b<d$, and $c<d<e$.
Then  ${\mathcal X}' =(X'_1, \ldots, X'_{11})$, where $X'_1=\emptyset$, $X'_2=\{a\}$,
$X'_3=\{a,c\}$, $X'_4=\{a\}$, $X'_5=\emptyset$, $X'_6=\{b\}$, $X'_7=\emptyset$, $X'_8=\{d\}$, $X'_9=\{d,e\}$,
$X'_{10}=\{d\}$, and $X'_{11}=\emptyset$
is a nice directed path-decomposition of width $1$  for $\g(Q)$ such that 
the introduce bags are ordered w.r.t.~$\alpha_Q$.
The execution of Algorithm {\sc Transform}  in Figure \ref{fig:algorithm3x} for the order of introduce
bags $t=(a,c,b,d,e)$ and the defined path in  state digraph  $\s(Q)$ is shown in Table \ref{TB2}. 
\end{example}

\begin{table}[hbtp]
\[\begin{array}{|r|lllll|l|c|c|c|}
\hline
j  & q_1 & ~ & q_2 &  ~ & q_3 & S & (v_0,\ldots,v_{15})\in\mathcal{P}(Q) &\ac(v_i) & |\ac(v_i)| \\
\hline
 0 & [a,a,d,e,d] & ~ & [c,b,b,d] & ~ & [c,c,d,e,d]& () & (0,0,0) &\emptyset & 0 \\
\hline
 1 & [a,d,e,d] & ~ & [c,b,b,d] &  ~ & [c,c,d,e,d] &  (a) & (1,0,0) & \{a\} & 1\\
   & [d,e,d] & ~ & [c,b,b,d] &  ~ & [c,c,d,e,d] &  (a) & (2,0,0) & \emptyset &0 \\
\hline
 2 & [d,e,d] & ~ & [b,b,d] &  ~ & [c,c,d,e,d] &  (a,c)  & (2,1,0)  & \{c\}&1\\
   & [d,e,d] & ~ & [b,b,d] &  ~ & [c,d,e,d] &  (a,c)  & (2,1,1)  & \{c\}&1 \\
   & [d,e,d] & ~ & [b,b,d] &  ~ & [d,e,d] &  (a,c)  & (2,1,2)  & \emptyset &0 \\
\hline
 3 & [d,e,d] & ~ & [b,d] & ~ & [d,e,d] &  (a,c,b) & (2,2,2)  &  \{b\} & 1 \\
   & [d,e,d] & ~ & [d] & ~ & [d,e,d] &  (a,c,b) & (2,3,2)  &  \emptyset &0\\
\hline
 4 & [e,d] & ~ & [d] & ~ & [d,e,d] &  (a,c,b,d) & (3,3,2) &  \{d\} &1\\
& [e,d] & ~ & [] & ~ & [d,e,d] &  (a,c,b,d) & (3,4,2) & \{d\}  &1\\
& [e,d] & ~ & [] & ~ & [e,d] &  (a,c,b,d) &  (3,4,3) &  \{d\} &1\\
\hline
 5 & [d] & ~ & [] & ~ & [e,d]  & (a,c,b,d,e) &  (4,4,3) &  \{d,e\}  & 2 \\
 & [] & ~ & [] & ~ & [e,d]  & (a,c,b,d,e) & (5,4,3)  & \{d,e\} &2\\
& [] & ~ & [] & ~ & [d]  & (a,c,b,d,e) &  (5,4,4)  & \{d\} &1\\
& [] & ~ & [] & ~ & []  & (a,c,b,d,e) &  (5,4,5)  &  \emptyset &0  \\
\hline
\end{array}\]
\caption{The execution of Algorithm   {\sc Transform} for $t=(a,c,b,d,e)$ and the 
list $Q = (q_1, q_2, q_3)$ of the sequences $q_1 = [a,a,d,e,d]$,
$q_2 = [c,b,b,d]$, and $q_3 = [c,c,d,e,c]$ from Example \ref{EX6}.}
\label{TB2} 
\end{table}

\begin{example} \label{EX7_2}
We consider the set $Q =\{q_1, q_2\}$ of the sequences $q_1 = [a,b,c,b,d,a]$
and $q_2 = [e,f,e,c]$ from Example \ref{EX5}.
Then ${\mathcal X} =(X_1, X_2, X_3)$, where $X_1=\{c,e,f\}$, $X_2=\{a,b,c\}$, and $X_3=\{a,b,d\}$  is a 
directed path-decomposition of width $2$  for $\g(Q)$ shown in Figure \ref{FFS}.

Our partial order $\alpha_Q$ has to fulfill $a<b<c<d$  and $e<f<c$.
Then ${\mathcal X}' =(X'_1, \ldots, X'_{13})$, where $X'_1=\emptyset$, $X'_2=\{e\}$,  
$X'_3=\{e,f\}$, $X'_4=\{c,e,f\}$, $X'_5=\{c,f\}$, $X'_6=\{c\}$, $X'_7=\{a,c\}$, $X'_8=\{a,b,c\}$, $X'_9=\{a,b\}$,
$X'_{10}=\{a,b,d\}$, $X'_{11}=\{b,d\}$, $X'_{12}=\{d\}$, and $X'_{13}=\emptyset$
is a nice directed path-decomposition of width $2$  for $\g(Q)$ such that 
the introduce bags are ordered w.r.t.~$\alpha_Q$.
The execution of Algorithm {\sc Transform}  in Figure \ref{fig:algorithm3x} for  the order of introduce
bags $t=(e,f,c,a,b,d)$ and the defined path in state digraph $\s(Q)$ is shown in Table \ref{TB2x}. 
\end{example}

\begin{table}[hbtp]
\[\begin{array}{|r|lll|l|c|c|c|}
\hline
j  & q_1 & ~ & q_2 &   S & (v_0,\ldots,v_r)\in\mathcal{P}(Q) & \ac(v_i) & |\ac(v_i)| \\
\hline
 0 &  [a,b,c,b,d,a] & ~ &[e,f,e,c] &  () & (0,0) & \emptyset &0 \\
\hline
 1 &  [a,b,c,b,d,a] & ~ &[f,e,c] &  (e) & (0,1) & \{e\} &1 \\
\hline
 2 &  [a,b,c,b,d,a] & ~ &[e,c] &  (e,f) & (0,2) & \{e,f\}&  2 \\
 &  [a,b,c,b,d,a] & ~ &[c] &  (e,f) & (0,3)  & \emptyset&  0 \\
\hline
3 &  [a,b,c,b,d,a] & ~ &[] &  (e,f,c) & (0,4)  & \{c\} &  1 \\
\hline
4 &  [b,c,b,d,a] & ~ &[] &  (e,f,c,a) & (1,4) & \{a,c\}&  2 \\
\hline
5 &  [c,b,d,a] & ~ &[] &  (e,f,c,a,b) & (2,4) & \{a,b,c\}&  3 \\
  &  [b,d,a] & ~ &[] &  (e,f,c,a,b) & (3,4) & \{a,b\}&  2 \\
&  [d,a] & ~ &[] &  (e,f,c,a,b) & (4,4) & \{a\}&  1 \\
\hline
6&  [a] & ~ &[] &  (e,f,c,a,b,d) & (5,4) & \{a,d\}&  2 \\
 &  [] & ~ &[] &  (e,f,c,a,b,d) & (6,4) & \emptyset &0 \\
\hline
\end{array}\]
\caption{The execution of Algorithm   {\sc Transform} for $t=(e,f,c,a,b,d)$ and the 
list $Q = (q_1, q_2)$ of the sequences $q_1 = [a,b,c,b,d,a]$
and $q_2 = [e,f,e,c]$ from Example \ref{EX5}.}
\label{TB2x} 
\end{table}

By Lemma \ref{lem-paths} and Lemma \ref{lem-paths2} we obtain the following result.

\begin{corollary}\label{th-ax1}
Given some set $Q$ of $k$ sequences, then
$$\dpw(\g(Q))=\min_{(v_0,\ldots,v_r)\in \mathcal{P}(Q)}\max_{1\leq i \leq r-1}|\ac(v_i)|-1.$$  
\end{corollary}

In order to apply Corollary \ref{th-ax1}
we consider some general digraph problem.
Let $G=(V,A,f)$ be a directed acyclic vertex-labeled graph. 
Function $f: V \to \ZZ$
assigns to every vertex $v \in V$  a value $f(v)$. Let
$s \in V$ and $t \in V$ be two vertices. For some vertex $v \in V$ and some
path $P=(v_1, \ldots, v_{\ell})$ with $v_1 = s$, $v_{\ell} = v$ and
$(v_i, v_{i+1}) \in A$ we define
$val_P(v) := \max_{u \in P} (f(u))$. Let ${\cal P}_s(v)$ denote the set
of all paths from vertex $s$ to vertex $v$. We define
$val(v) := \min_{P \in {\cal P}_s(v)} (val_P(v))$.
Then it holds:
$$val(v) = \max\{ f(v), \min_{u \in N^-(v)}(val(u)) \}.$$
%
%
In Figure \ref{fig:algorithm3dp} we give a dynamic 
programming solution to compute all the values of 
$val(v)$, $v\in V$ in time $\bigo(|V|+|A|)$.
This is possible, since the graph is directed and acyclic.

\begin{figure}[h]
\hrule
{\strut\footnotesize \bf Algorithm {\sc Dynamic Programming} (DP)}
\hrule
\begin{tabbing}
xxxx \= xxxx \= xxxx \= xxxx \= \kill
$val[s] := f(s)$ \\
for every vertex $v \neq s$ in order of $topol$ do $\{$\\
\> $val[v] := \infty$ \\
\> for every $u \in N^-(v)$ do \\
\>\> if ($val[u] < val[v]$) \\
\>\>\> $val[v] := val[u]$;         ~~~~~\textcolor{red}{$val[v]=\min_{u \in N^-(v)}val(u)$}\\
\> if ($val[v] < f(v)$) \\
\>\> $val[v] := f(v)$              ~~~~~~~~~~~~~~~\textcolor{red}{$val[v]=\max\{ f(v), \min_{u \in N^-(v)}val(u) \}$}  \\
\> $\}$
\end{tabbing}
\hrule
\caption{Computing a path where the maximum value is as small as possible.}
\label{fig:algorithm3dp}
\end{figure}

Corollary \ref{th-ax1} allows us to apply the algorithm given in Figure \ref{fig:algorithm3dp} 
on  $\s(Q)$ to compute the directed path-width of $\g(Q)$.


\begin{theorem}\label{th-axx}
Given some set $Q$, such that  $\g(Q)\in S_{k,\ell}$ for some $\ell \geq 1$, then the 
directed path-width of $\g(Q)$ and a directed path-decomposition 
can be computed in time 
$\bigo(k\cdot (1+\max_{1\leq i \leq k}n_i)^{k})$.
\end{theorem}

\begin{proof}
Let $Q$ be some set, such that $\g(Q)\in S_{k,\ell}$. 
The state digraph $\s(Q)$  has at most $(1+\max_{1\leq i \leq k}n_i)^k$
vertices and  can be found in time $\bigo(k \cdot (1+\max_{1\leq i \leq k}n_i)^k)$ from $Q$.
By Corollary \ref{th-ax1} the directed path-width of $\g(Q)$ can be computed by considering 
all paths from the initial state to the final state in $\s(Q)$. 
This can be done by the algorithm given in Figure \ref{fig:algorithm3dp} 
on $\s(Q)=(V,A)$ using $f(v)=|\ac(v)|$, $v\in V$, $s$ as the initial state, and $t$ 
as the final state.
Since every vertex of the state digraph has at most $k$ outgoing arcs we have 
$\bigo(|V|+|A|)\subseteq \bigo(k\cdot(1+\max_{1\leq i \leq k}n_i)^{k})$.
Thus we can compute
an optimal path in $\s(Q)$ in time $\bigo(k\cdot(1+\max_{1\leq i \leq k}n_i)^{k})$.
\end{proof}

These results lead to an XP-algorithm
for directed path-width w.r.t.~number of sequences $k$ needed to define the input graph.
This implies that for each constant $k$, it is decidable in polynomial time whether 
for a given set $Q$ on at most $k$ sequences the
digraph $\g(Q)$ has directed path-width at most $w$.
If we know that some graph can be defined by one sequence, we can 
find this in linear time (Proposition \ref{le-find}).
This implies that for each constant $k$, it is decidable in polynomial time whether 
for a digraph $G$, which is given by the union of at most $k$ 
many $\{\co(2\overrightarrow{P_2}),2\overleftrightarrow{K_1},\overrightarrow{C_3},D_4\}$-free
digraphs, digraph $G$ 
has directed path-width at most $w$.

\section{Conclusions and Outlook}\label{sec-concl}

There are several interesting open questions. 
\begin{inparaenum}[\bf(a)]
\item 
In Section \ref{sec-dpw-a} we have shown XP-algorithms for the directed path-width problem
w.r.t.~parameter $k$ (number of sequences). It remains open whether there
are also FPT-algorithms for this parameter. The same question is open for the
standard parameter.
\item How (efficiently) can we also compute the directed tree-width of sequence digraphs?
\item Can Remark \ref{cor-union-transi} be generalized to a (larger) number
of transitive digraphs?
\item 
Does the hardness of Proposition \ref{hard-dpw-deg} also hold for
$c_Q \in \{2,3,4\}$ and for $d_Q = 2$?
\item By  Theorem \ref{s11} one can decide in polynomial time
whether a given digraph belongs to the class $S_{1,1}$. By Lemma
\ref{unionsk1} this is also possible for every class $S_{k,1}$. Furthermore
Theorem \ref{s12} allows to recognize  the classes $S_{1,2}$ in polynomial
time.
It remains to consider this problem for the classes $S_{k,\ell}$ for $k\geq 2$ and $1<\ell\leq 2k$.
\item In order to apply Corollary \ref{th-ax1}
for some given digraph $G$ we need a set $Q$ such that $\g(Q)=G$.
Can we find a set $Q$ with a smallest number of sequences 
in polynomial time?
\item Can we find a {\em simple} set $Q$ with a smallest number of sequences 
in polynomial time?
\item By its equivalence to directed vertex separation number \cite{YC08} 
directed path-width is a vertex ordering problem (layout problem).
This motivates to solve similar problems such as directed cut-width or directed
feedback arc set  (see \cite{BFKKT12} for a survey)
on sequence digraphs. Is there even a general approach to solve  vertex ordering problem 
on  sequence digraphs (see \cite{BFKKT12} for a general approach on undirected
graphs)?
\end{inparaenum}

\section{Acknowledgements} \label{sec-a}

The work of the second author was supported by the German Research 
Association (DFG) grant  GU 970/7-1.



\begin{thebibliography}{10}

\bibitem{ACP87}
S.~Arnborg, D.G. Corneil, and A.~Proskurowski.
\newblock Complexity of finding embeddings in a $k$-tree.
\newblock {\em SIAM Journal of Algebraic and Discrete Methods}, 8(2):277--284,
  1987.

\bibitem{BG09}
J.~Bang-Jensen and G.~Gutin.
\newblock {\em Digraphs. {T}heory, {A}lgorithms and {A}pplications}.
\newblock Springer-Verlag, Berlin, 2009.

\bibitem{Bar06}
J.~Bar{\'a}t.
\newblock {D}irected pathwidth and monotonicity in digraph searching.
\newblock {\em Graphs and Combinatorics}, 22:161--172, 2006.

\bibitem{BGR97}
D.~Bechet, P.~de~Groote, and C.~Retor\'e.
\newblock A complete axiomatisation of the inclusion of series-parallel partial
  orders.
\newblock In {\em Rewriting Techniques and Applications}, volume 1232 of {\em
  LNCS}, pages 230--240. Springer-Verlag, 1997.

\bibitem{Bod96}
H.L. Bodlaender.
\newblock A linear-time algorithm for finding tree-decompositions of small
  treewidth.
\newblock {\em SIAM Journal on Computing}, 25(6):1305--1317, 1996.

\bibitem{Bod97}
H.L. Bodlaender.
\newblock Treewidth: Algorithmic techniques and results.
\newblock In {\em Proceedings of Mathematical Foundations of Computer Science},
  volume 1295 of {\em LNCS}, pages 29--36. Springer-Verlag, 1997.

\bibitem{Bod98}
H.L. Bodlaender.
\newblock A partial $k$-arboretum of graphs with bounded treewidth.
\newblock {\em Theoretical Computer Science}, 209:1--45, 1998.

\bibitem{BFKKT12}
H.L. Bodlaender, F.V. Fomin, A.M.C.A. Koster, D.~Kratsch, and D.M. Thilikos.
\newblock A note on exact algorithms for vertex ordering problems on graphs.
\newblock {\em Theory of Computing Systems}, 50(3):420--432, 2012.

\bibitem{Boe15}
D.~Boeckner.
\newblock Oriented threshold graphs.
\newblock {\em ACM Computing Research Repository (CoRR)}, abs/1511.01008:12
  pages, 2015.

\bibitem{CFS12}
M.~Chudnovsky, A.O. Fradkin, and P.D. Seymour.
\newblock Tournament immersion and cutwidth.
\newblock {\em Journal of Combinatorial Theory, Series B}, 102(1):93--101,
  2012.

\bibitem{CP06}
C.~Crespelle and C.~Paul.
\newblock Fully dynamic recognition algorithm and certificate for directed
  cographs.
\newblock {\em Discrete Applied Mathematics}, 154(12):1722--1741, 2006.

\bibitem{Gou12}
R.~Gould.
\newblock {\em Graph Theory}.
\newblock Dover, 2012.

\bibitem{GRR18}
F.~Gurski, C.~Rehs, and J.~Rethmann.
\newblock Directed pathwidth of sequence digraphs.
\newblock In {\em Proceedings of the International Conference on Combinatorial
  Optimization and Applications (COCOA)}, LNCS. Springer-Verlag, 2018.
\newblock to appear.

\bibitem{GRW16b}
F.~Gurski, J.~Rethmann, and E.~Wanke.
\newblock On the complexity of the {FIFO} stack-up problem.
\newblock {\em Mathematical Methods of Operations Research}, 83(1):33--52,
  2016.

\bibitem{Gus93}
J.~Gusted.
\newblock On the pathwidth of chordal graphs.
\newblock {\em Discrete Applied Mathematics}, 45(3):233--248, 1993.

\bibitem{Har69}
F.~Harary.
\newblock {\em Graph Theory}.
\newblock Addison-Wesley Publishing Company, Massachusetts, 1969.

\bibitem{JRST01}
T.~Johnson, N.~Robertson, P.D. Seymour, and R.~Thomas.
\newblock Directed tree-width.
\newblock {\em Journal of Combinatorial Theory, Series B}, 82:138--155, 2001.

\bibitem{KF79}
T.~Kashiwabara and T.~Fujisawa.
\newblock {NP}-completeness of the problem of finding a minimum-clique-number
  interval graph containing a given graph as a subgraph.
\newblock In {\em Proceedings of the International Symposium on Circuits and
  Systems}, pages 657--660, 1979.

\bibitem{KP18}
S.V. Kitaev and A.V. Pyatkin.
\newblock Word-representable graphs: a survey.
\newblock {\em Journal of Applied and Industrial Mathematics}, 12(2):278--296,
  2018.

\bibitem{KKKTT16}
K.~Kitsunai, Y.~Kobayashi, K.~Komuro, H.~Tamaki, and T.~Tano.
\newblock Computing directed pathwidth in ${O}(1.89^n)$ time.
\newblock {\em Algorithmica}, 75:138--157, 2016.

\bibitem{KKT15}
K.~Kitsunai, Y.~Kobayashi, and H.~Tamaki.
\newblock On the pathwidth of almost semicomplete digraphs.
\newblock In {\em Proceedings of the Annual European Symposium on Algorithms},
  volume 9294 of {\em LNCS}, pages 816--827. Springer-Verlag, 2015.

\bibitem{KBMK93}
T.~Kloks, H.~Bodlaender, H.~M{\"u}ller, and D.~Kratsch.
\newblock Computing treewidth and minimum fill-in: {A}ll you need are the
  minimal separators.
\newblock In {\em Proceedings of the Annual European Symposium on Algorithms},
  volume 726 of {\em LNCS}, pages 260--271. Springer-Verlag, 1993.

\bibitem{Kob15}
Y.~Kobayashi.
\newblock Computing the pathwidth of directed graphs with small vertex cover.
\newblock {\em Information Processing Letters}, 115(2):310--312, 2015.

\bibitem{MS88}
B.~Monien and I.H. Sudborough.
\newblock Min cut is {NP}-complete for edge weighted trees.
\newblock {\em Theoretical Computer Science}, 58:209--229, 1988.

\bibitem{Nag12}
H.~Nagamochi.
\newblock Linear layouts in submodular systems.
\newblock In {\em Proceedings of the International Symposium on Algorithms and
  Computation}, volume 7676 of {\em LNCS}, pages 475--484. Springer-Verlag,
  2012.

\bibitem{RS83}
N.~Robertson and P.D. Seymour.
\newblock Graph minors {I}. {E}xcluding a forest.
\newblock {\em Journal of Combinatorial Theory, Series B}, 35:39--61, 1983.

\bibitem{YC08}
B.~Yang and Y.~Cao.
\newblock Digraph searching, directed vertex separation and directed pathwidth.
\newblock {\em Discrete Applied Mathematics}, 156(10):1822--1837, 2008.

\end{thebibliography}

\end{document}